\documentclass[onecolumn]{IEEEtran}
\usepackage{amsmath,amssymb,enumerate,theorem}
\usepackage{bm,color,xcolor}
\usepackage{cite}
\usepackage{graphicx}
\usepackage{float}
\usepackage{extarrows}
\usepackage{multirow}
\usepackage{hyperref}
\usepackage{cleveref}
\usepackage{soul}

\newtheorem{theorem}{Theorem}[section]
\newtheorem{lemma}{Lemma}[section]

\newtheorem{fact}{Fact}[section]
\newtheorem{proposition}{Proposition}[section]
\newtheorem{corollary}{Corollary}[section]
\newtheorem{definition}{Definition}[section]

\newtheorem{example}{Example}[section]
\newtheorem{remark}{Remark}[section]
\newtheorem{construction}{Construction}[section]

\newcommand\nc\newcommand
\nc{\cA}{\mathcal{A}}\nc{\cB}{\mathcal{B}}\nc{\cC}{\mathcal{C}}\nc{\cD}{\mathcal{D}}
\nc{\cE}{\mathcal{E}}\nc{\cF}{\mathcal{F}}\nc{\cG}{\mathcal{G}}\nc{\cH}{\mathcal{H}}
\nc{\cI}{\mathcal{I}}\nc{\cJ}{\mathcal{J}}\nc{\cK}{\mathcal{K}}\nc{\cL}{\mathcal{L}}
\nc{\cM}{\mathcal{M}}\nc{\cN}{\mathcal{N}}\nc{\cO}{\mathcal{O}}\nc{\cP}{\mathcal{P}}
\nc{\cQ}{\mathcal{Q}}\nc{\cR}{\mathcal{R}}\nc{\cS}{\mathcal{S}}\nc{\cT}{\mathcal{T}}
\nc{\cU}{\mathcal{U}}\nc{\cV}{\mathcal{V}}\nc{\cW}{\mathcal{W}}\nc{\cX}{\mathcal{X}}
\nc{\cY}{\mathcal{Y}}\nc{\cZ}{\mathcal{Z}}

\nc{\bba}{\mathbf{a}}\nc{\bbb}{\mathbf{b}}\nc{\bbc}{\mathbf{c}}\nc{\bbd}{\mathbf{d}}
\nc{\bbe}{\mathbf{e}}\nc{\bbf}{\mathbf{f}}\nc{\bbg}{\mathbf{g}}\nc{\bbh}{\mathbf{h}}
\nc{\bbi}{\mathbf{i}}\nc{\bbj}{\mathbf{j}}\nc{\bbk}{\mathbf{k}}\nc{\bbl}{\mathbf{l}}
\nc{\bbm}{\mathbf{m}}\nc{\bbn}{\mathbf{n}}\nc{\bbo}{\mathbf{o}}\nc{\bbp}{\mathbf{p}}
\nc{\bbq}{\mathbf{q}}\nc{\bbr}{\mathbf{r}}\nc{\bbs}{\mathbf{s}}\nc{\bbss}{\mathbf{s}^{\star}}\nc{\bbt}{\mathbf{t}}\nc{\bbtt}{\mathbf{t}^{\star}}
\nc{\bbu}{\mathbf{u}}\nc{\bbv}{\bm{v}}\nc{\bbw}{\mathbf{w}}\nc{\bfx}{\mathbf{x}}
\nc{\bby}{\mathbf{y}}\nc{\bbz}{\mathbf{z}}

\nc{\bbA}{\mathbf{A}}\nc{\bbB}{\mathbf{B}}\nc{\bbC}{\mathbf{C}}\nc{\bbD}{\mathbf{D}}
\nc{\bbE}{\mathbf{E}}\nc{\bbF}{\mathbf{F}}\nc{\bbG}{\mathbf{G}}\nc{\bbH}{\mathbf{H}}
\nc{\bbI}{\mathbf{I}}\nc{\bbJ}{\mathbf{J}}\nc{\bbK}{\mathbf{K}}\nc{\bbL}{\mathbf{L}}
\nc{\bbM}{\mathbf{M}}\nc{\bbN}{\mathbf{N}}\nc{\bbO}{\mathbf{O}}\nc{\bbP}{\mathbf{P}}
\nc{\bbQ}{\mathbf{Q}}\nc{\bbR}{\mathbf{R}}\nc{\bbS}{\mathbf{S}}\nc{\bbT}{\mathbf{T}}
\nc{\bbU}{\mathbf{U}}\nc{\bbV}{\mathbf{V}}\nc{\bbW}{\mathbf{W}}\nc{\bfX}{\mathbf{X}}
\nc{\bbY}{\mathbf{Y}}\nc{\bbZ}{\mathbf{Z}}

\nc{\sA}{\mathsf{A}}\nc{\sB}{\mathsf{B}}\nc{\sC}{\mathsf{C}}\nc{\sD}{\mathsf{D}}
\nc{\sE}{\mathsf{E}}\nc{\sF}{\mathsf{F}}\nc{\sG}{\mathsf{G}}\nc{\sH}{\mathsf{H}}
\nc{\sI}{\mathsf{I}}\nc{\sJ}{\mathsf{J}}\nc{\sK}{\mathsf{K}}\nc{\sL}{\mathsf{L}}
\nc{\sM}{\mathsf{M}}\nc{\sN}{\mathsf{N}}\nc{\sO}{\mathsf{O}}\nc{\sP}{\mathsf{P}}
\nc{\sQ}{\mathsf{Q}}\nc{\sR}{\mathsf{R}}\nc{\sS}{\mathsf{S}}\nc{\sT}{\mathsf{T}}
\nc{\sU}{\mathsf{U}}\nc{\sV}{\mathsf{V}}\nc{\sW}{\mathsf{W}}\nc{\sX}{\mathsf{X}}
\nc{\sY}{\mathsf{Y}}\nc{\sZ}{\mathsf{Z}}

\newcommand{\mathset}[1]{\left\{#1\right\}}
\newcommand{\multiset}[1]{\left\{\left\{#1\right\}\right\}}
\newcommand{\abs}[1]{\left|#1\right|}
\newcommand{\ceilenv}[1]{\left\lceil #1 \right\rceil}
\newcommand{\floorenv}[1]{\left\lfloor #1 \right\rfloor}
\newcommand{\parenv}[1]{\left( #1 \right)}
\newcommand{\sparenv}[1]{\left[ #1 \right]}

\nc{\set}[1]{\llbracket #1 \rrbracket}

\newcommand{\wth}[1]{\mathsf{wt}_H\left(#1\right)}

\newcommand{\ed}[2]{d_{#1\text{-}H}\left(#2\right)}
\newcommand{\eS}[2]{\mathcal{S}_{#1\text{-}H}\left(#2\right)}
\newcommand{\eB}[2]{\mathcal{B}_{#1\text{-}H}\left(#2\right)}
\newcommand{\vS}[2]{S_{#1\text{-}H}\left(#2\right)}
\newcommand{\vB}[2]{B_{#1\text{-}H}\left(#2\right)}

\begin{document}

\title{\textbf{More on Codes for \\ Combinatorial Composite DNA}}
\author{\textbf{\IEEEauthorblockN{Zuo Ye} \and\hspace*{.5in} \IEEEauthorblockN{Omer Sabary} \and\hspace*{.5in} \IEEEauthorblockN{Ryan Gabrys} \and\hspace*{.5in} \IEEEauthorblockN{Eitan Yaakobi} \and\hspace*{.5in} \IEEEauthorblockN{Ohad Elishco}}
}

\maketitle

{\renewcommand{\thefootnote}{}\footnotetext{

\vspace{-.2in}
 
\noindent\rule{1.5in}{.4pt}

{ 

O. Elishco and Z. Ye are with the School of Electrical and Computer Engineering, Ben-Gurion University of the Negev, 
Be'er Sheva, Israel. Email: elishco@gmail.com. 
O. Elishco and Z. Ye research was supported in part by the Israel Science Foundation (Grant No. 1789/23). 
R. Gabrys is with the University of California, San Diego, California, USA. Email: rgabrys@ucsd.edu. 
R. Gabrys research was supported in part by NSF Grant CCF2212437. O. Sabary and E. Yaakobi are with the Henry and Marilyn Taub Faculty of Computer Science, Technion - Israel Institute of Technology, Haifa 3200003, Israel 
  (e-mail: \texttt{\{omersabary,\linebreak yaakobi\}@cs.technion.ac.il}). 
  O. Sabary and E. Yaakobi have received funding for this project from the European Union (DiDAX, 101115134). Views and opinions expressed are, however, those of the author(s) only and do not necessarily reflect those of the European Union or the European Research Council Executive Agency. Neither the European Union nor the granting authority can be held responsible for them.
}
}}

\begin{abstract}
In this paper, we focus on constructions of unique-decodable/list-decodable on the recently studied $(t,e)$-composite-asymmetric error-correcting codes ($(t,e)$-CAECCs). Let $\cX$ be an $m\times n$ binary matrix, in which each row has Hamming weight $w$. When at most $t$ rows of $\cX$ suffer from errors and in each of these erroneous rows, there are at most $e$ $1 \to 0$ errors, we say that  a $(t,e)$-composite-asymmetric-error occurs in $\cX$.

For general $m,n,w,t,e$, we propose new constructions of $(t,e)$-CAECCs with redundancy at most $(t-1)\log(m)+O(1)$, where $O(1)$ is a number independent of the code-length $m$. In particular, this gives a class of $(2,e)$-CAECCs that are optimal in terms of their redundancy. 
When $m$ is a prime power, the redundancy can be further reduced to $(t-1)\log(m)-O(\log(m))$.
To further increase the size of these codes, we introduce a combinatorial object called a weak $B_e$-sets. When $e=w$, we show an efficient way to encode/decode our codes.

At last, we investigate how much we can gain if we relax the requirement of uniquely decoding to list-decoding. It is shown that when the list size is $t!$ or an exponential function of $t$, there are list-decodable $(t,e)$-CAECCs with constant redundancy. When the list size is two, we show that there are list-decodable $(3,2)$-CAECCs with redundancy $\log(m)+O(1)$.
\end{abstract}

\section{Introduction}
The digital data in the current information age is undergoing rapid growth. According to the International Data Corporation (IDC), the global datasphere is predicted to expand to 175 zettabytes by 2025 \cite{IDC2018}. This massive accumulation of digital data has prompted the search for alternative storage methods. One promising solution is DNA-based storage, which offers high information density, long-term stability, and robustness \cite{Yazdi2015TMBMC}.
In classical DNA-based data storage systems, binary digital data is firstly converted into sequences (DNA sequences) over the standard DNA alphabet $\mathset{A,C,G,T}$ (since there are four DNA bases $A$, $C$, $G$ and $T$). Then according to these quaternary sequences, DNA molecules called strands are generated via a biochemical process called \emph{DNA synthesis}. After that, the synthesized strands are stored together in a suitable container. To retrieve the original binary data, a process termed \emph{DNA sequencing} is applied to stored strands to read back the quaternary sequences, which will be further decoded back into their corresponding binary representations.

Although the feasibility of DNA-based storage has been demonstrated experimentally in many works (see Table 1 in \cite{Leon2019NB} for a summary), the prohibitively high cost of DNA-based storage keeps it far from commercial applications.
The expense of DNA storage is influenced by both the cost of DNA synthesis and the price of data retrieval via DNA sequencing
It is reported that synthesis cost per position is approximately $80,000$ times more expensive than sequencing cost per position \cite{YANIV2017Science}. Therefore, there is a strong motivation to reduce the synthesis cost. One of the proposed solutions is to maximize the logical density, i.e., the number of written bits per cycle of synthesis (bits/cycle) \cite{Yeongjae2019SR,Leon2019NB}.
In classical DNA-based storage systems, synthesizing a strand is done base-by-base, meaning one of the four bases is added to the growing strand with each cycle of synthesis. Consequently, the logical density is theoretically bounded by $\log 4 = 2$ bits per cycle. Furthermore, since the processes of synthesis, storage, and sequencing are prone to errors, error-correcting codes need to be integrated into the storage system, which further reduces the achievable logical density. 

It should be noted that standard DNA synthesis and sequencing methodologies create a high number of strands for each of the designed DNA sequences. Recently, inspired by the inherent multiplicity of the created DNA strands, Choi \textit{et al.} \cite{Yeongjae2019SR} and Anavy \textit{et al.} \cite{Leon2019NB} suggested a novel way to extend the standard DNA alphabet, by the addition of composite DNA letters. A composite letter is a combination of standard DNA letters ($A$, $C$, $G$ and $T$) with predetermined ratio $\sigma=\parenv{\sigma_A,\sigma_C,\sigma_G,\sigma_T}$, where $\sigma_A+\sigma_C+\sigma_G+\sigma_T=k$ and $k$ is the resolution parameter of the composite letter. The alphabet consisting of all composite letters of resolution $k$ is denoted by $\Phi_k$. A subset of $\Phi_k$ is called a composite alphabet. Since the cardinality of $\Phi_k$ is $\binom{k+3}{3}$, the theoretical limit of logical density can be increased beyond $2$ bits/cycle. As examples, the authors of \cite{Leon2019NB} experimentally achieved logical densities of $1.96$ bits/cycle when storing a $6.42$ MB file by using a six-letter subset of $\Phi_2$ as the alphabet, and $4.29$ bits/cycle when storing a $22.5$-byte file by using $\Phi_3$ as the alphabet. The channel of composite DNA has also been studied from a coding theory perspective in several works~\cite{Wenkai2022ICC,Frederik2024ISIT,Wenkai202405}.

Although the use of composite alphabets contributes higher logical densities and thus reduces the overall cost of a DNA-based storage system, it is a complex task to synthesize and/or read a composite letter with given ratio $\sigma$ \cite{Yeongjae2019SR,Leon2019NB}. Inspired by these two works and motif-based DNA storage \cite{Roquet2021,Eugenio2022}, Preuss \textit{et al.}~\cite{Press2024SR} and Yan \textit{et al.} \cite{Yan2023SR} extended the idea of composite letters to be based on on \emph{motifs}, also known as \emph{shortmers}, which are short oligonucleotide sequences. They used a set of motifs, instead of a combination of DNA bases with a given ratio, as a composite letter. Due to this reason, we also call a composite letter in their scheme a \textit{combinatorial composite letter}. They experimentally achieved a logical density of $84$ bits/cycle when storing a $80$-bit text. In their storage systems, the ratio of different motifs in a composite letter was not taken into account. As a result, the synthesis and reading are much less complex than those of schemes proposed in \cite{Yeongjae2019SR,Leon2019NB}.

Preuss \textit{et al.} \cite{Press2024SR} presented a novel DNA storage scheme in which a combinatorial composite letter is composed of $w$ DNA $k$-mers (motifs of length $k$) from a fixed set $\Omega$ of $n$ prefabricated DNA $k$-mers. Therefore, the theoretical limit of logical density is $\log\binom{n}{w}$ bits/cycle. This DNA storage system is prone to several types of errors, among which is the missing of one or more $k$-mers in a combinatorial composite letter. Since a combinatorial composite letter is a subset of $\Omega$ of size $w$, it can also be represented as a binary vector of Hamming weight $w$, that is, the indicator vector of this subset. Then a sequence of $m$ composite letters is equivalent to an $m\times n$ binary matrix, where each row represents a composite letter and thus has Hamming weight $w$ (the set of all such matrices is denoted by $\Sigma_w^{m\times n}$). Thus, when a $k$-mer is missing in a composite letter, it results in a decrease of the Hamming weight in the row corresponding to that composite letter. 

Suppose there are at most $t$ erroneous composite letters, and each composite letter loses at most $e$ $k$-mers. This type of error is termed a $(t,e)$-composite-asymmetric error ($(t,e)$-CAE). Sabary \textit{et al.} \cite{Omer2024ISIT} initiated the study of codes correcting such errors. When $m$ is at most $p^e$, where $p$ is the smallest prime between $n$ and $2n$, a class of codes with redundancy at most $et\log(p)$ was constructed in \cite{Omer2024ISIT}. When $e=1$, $n$ is a prime and $m\le n$, they presented an efficient encoder/decoder for the constructed codes. They also defined a metric, called the $2e$-Hamming distance, on $\Sigma_w^{m\times n}$. They showed that a subset of $\Sigma_w^{m\times n}$ is a $(t,e)$-CAE correcting code if and only if its minimum $2e$-Hamming distance is at least $t+1$. With this conclusion, they derived a sphere-packing upper bound on the maximum size of $(t,e)$-CAE correcting codes. For other works related to the storage scheme proposed in \cite{Press2024SR}, we refer to \cite{Preuss2024,Cohen2024ISIT,Roman202406}.

Following the work of Sabary \textit{et al.} \cite{Omer2024ISIT}, this paper aims to further explore $(t,e)$-CAE correcting codes. We first improved the sphere-packing bound in \cite{Omer2024ISIT} to a Johnson-type upper bound when $t$ is odd. To correct a $(t,e)$-CAE, one has to identify error positions in each row of the corrupted codeword (which is a matrix). To that end, the idea in \cite{Omer2024ISIT} is to firstly calculate $e$ syndromes of each row of the matrix. Then those $e$ syndromes of the same row are encoded into a symbol in the finite field $\mathbb{F}_{p^e}$. Thus, the matrix is encoded into a vector in $\mathbb{F}_{p^e}^m$. Now, the code presented in \cite{Omer2024ISIT} is a subset of matrices whose corresponding vectors belong to the same coset of a linear $t$-erasure correcting code in $\mathbb{F}_{p^e}^m$. In this paper, instead of encoding the $e$ syndromes in a symbol into $\mathbb{F}_{p^e}$, we encode them into an integer between $0$ and $p^e-1$. With this idea, we succeed in constructing codes with more flexible parameters. Specifically, for general $t\ge2$, we give codes correcting a $(t,e)$-composite-asymmetric error with redundancy at most $(t-1)\log(m)+O(1)$, where $O(1)$ is a number depending only on $n,e$ and $t$. Therefore, our code outperforms the one given in \cite{Omer2024ISIT} when $n,e$ and $t$ are fixed while $m$ is large. When $m$ is a prime power, the redundancy can be further reduced to $(t-1)\log(m)-O(\log(m))$.
In addition to gains in redundancy, our constructions of codes and decoding processes only involve modular operations of integers, which are simpler than operations of elements in the field $\mathbb{F}_{p^e}$. To further increase code size, we define the weak $B_e$ set, which can be applied to identify error positions in a row using only one syndrome. As a byproduct, an efficient encoding/decoding algorithm is obtained for the case $e=w$.

We also studied list-decodable codes. A code is called $(t,e)$-list decodable with list size $L$, if the intersection of error balls of any $L+1$ codewords is empty. Firstly, we show that a $(d-1,e)$-CAE correcting code with minimum $2e$-Hamming distance $d$ is $(t,e)$-list decodable with a list size upper bounded by a function of $d$ and $t$. As a corollary, we derive an Elias-type upper bound on the maximum size of $(d-1,e)$-CAE correcting codes. Regarding constructions, we show that when the list size is $t!$ or an exponential function of $t$, there are $(t,e)$-CAE list-decodable codes with redundancy independent of $m$. When the list size is two, we show that there is a $(3,2)$-CAE list-decodable code with redundancy $\log(m)+O(1)$.

This paper is organized as follows. In \Cref{sec_pre}, some necessary terminologies and auxiliary results are introduced. \Cref{sec_unique} and \Cref{sec_list} are devoted to bounds and constructions of $(t,e)$-CAE correcting codes and $(t,e)$-CAE list-decodable codes, respectively. At last, \Cref{sec_conclusion} concludes this paper.

\section{Preliminaries}\label{sec_pre}
For an integer $n\ge1$, denote $[n]\triangleq\mathset{0,1,\ldots,n-1}$. For an integer $Q>1$, let $\cA_Q$ be the alphabet $\mathset{0,1,\ldots,Q-1}$. The set of length-$n$ sequences over $\cA_Q$ is denoted by $\cA_Q^n$. In this paper, we represent sequences and vectors with bold letters, such as $\bm{x}$ and $\bm{y}$. Symbols in a sequence of length $n$ are indexed by $[n]$. For $i\in[n]$, the $i$-th symbol in sequence $\bm{x}\in\cA_Q^n$ is denoted by $x_i$. In other words, $\bm{x}=x_0x_1\cdots x_{n-1}$. The Hamming weight of a sequence $\bm{x}$ is denoted by $\wth{\bm{x}}$. The cardinality of a finite set $S$ is denoted by $\abs{S}$.

\subsection{Combinatorial DNA}
Let $\cA_4$ be the DNA quaternary alphabet, where each symbol can be represented by $\{A, C, G, T\}$ and let $\ell \in N$ be the shortmer length. We let $\cS = \{ \textbf{s}_0, \textbf{s}_1, \ldots, \textbf{s}_{n-1} \}$ be a set of $n>1$ \emph{shortmers}, $\textbf{s}_i \in \cA_4^\ell$ for $i\in [n]$. For $w<n$, the set of all subsets of $\cS$ of size $w$ defines the \emph{$w$-combinatorial composite alphabet of $\cS$}, where each \emph{combinatorial composite symbol} $\bfx_i^{\cS}$, for $i\in [\binom{n}{w}]$ is a \emph{set} of $w$ different shortmers chosen from the shortmers set $\cS$.
For simplicity, as was done in~\cite{Omer2024ISIT}, the $w$-composite symbols are abstracted as length-$n$ binary vectors of weight $w$ in which every bit indicates whether a shortmer in $\cS$ belongs to the set. The example below demonstrates one of the combinatorial symbols that were used in~\cite{Preuss2024}. 

\begin{small}
\begin{example} In~\cite{Press2024SR}, the authors used the following parameters {$\ell=3$}, $n=16$ and $w=5$,  

\vspace{-2.5ex}\begin{small}
\begin{align*}\cS =
\begin{Bmatrix}
\textbf{s}_0  = AAT, \textbf{s}_1 = ACA, \textbf{s}_2 = ATG, \textbf{s}_3 = AGC, \\ \textbf{s}_4 = TAA, \textbf{s}_5 = TCT, \textbf{s}_6 = TTC, \textbf{s}_7 = TGG, \\ \textbf{s}_8 = GAG, \textbf{s}_{9} = GCC, \textbf{s}_{10} = GTT, \textbf{s}_{11} = GGA, \\ \textbf{s}_{12} = CAC, \textbf{s}_{13} = CCG, \textbf{s}_{14} = CTA, \textbf{s}_{15} = CGT
\end{Bmatrix} .
\end{align*}
\end{small}


The set $\cS$ was selected as a code with Hamming distance  $d=2$. In this setup, an example of the first composite symbol of the alphabet is $\bfx_{0}^\cS = (1, 1, 1, 1, 1, 0, 0, 0, 0, 0, 0, 0, 0, 0, 0, 0)$, which represents the set consisting of the shortmers $\{ \textbf{s}_0, \textbf{s}_1, \textbf{s}_2, \textbf{s}_3, \textbf{s}_{4}\}$. 
\end{example}
\end{small}

Following the definitions above, a sequence of length $m>1$ over the $w$-combintorial composite alphabet of $\cS$ can be represented as an $m \times n$ binary array $\mathcal{X}$, in which each row is of weight $w$. The set of such matrices is denoted by $\Sigma_{w}^{m\times n}$. Rows and columns in an $m\times n$ matrix $\cX$ (not necessarily in $\Sigma_{w}^{m\times n}$) are indexed by $[m]$ and $[n]$, respectively. The $i$-th row of $\cX$ is denoted by $\cX_i$.

\subsection{Model description}
In this paper, we examine asymmetric errors, in which bits can only change from 1 to 0. The next definition formalizes some of these ideas. 
\begin{definition}
Let $1\le e\le w$. Given a sequence $\bm{x}\in\cA_2^n$ with Hamming weight $w$, we say that a sequence $\bm{y}\in\cA_2^n$ is obtained from $\bm{x}$ by $e$ \emph{asymmetric errors}, if $\wth{\bm{y}}=w-e$ and $y_i\le x_i$ for each $i\in[n]$.
\end{definition}

\begin{definition}\cite{Omer2024ISIT}
  Let $1\le e\le w<n$ and $1\le t<m$. Given a matrix $\cX\in\Sigma_w^{m\times n}$, we say an $m\times n$ matrix $\cY$ is obtained from $\cX$ by a $(t,e)$-\emph{composite-asymmetric-error} (or $(t,e)$-CAE for short), if there are at most $t$ rows in $\cX$ that have experienced at most $e$ composite-asymmetric errors.
\end{definition}

\begin{definition}\cite{Omer2024ISIT}
  A nonempty set $\cC\subseteq\Sigma_w^{m\times n}$ is called a $(t,e)$-composite-asymmetric error-correcting code (or $(t,e)$-CAECC for short), if it can correct any $(t,e)$-CAE. Such a code will also be referred as an $\sparenv{m,(n,w);t,e}$-composite code.
\end{definition}

Let $\cC$ be an $\sparenv{m,(n,w);t,e}$-composite code. The size of $\cC$ is denoted by $\abs{\cC}$ and the redundancy of $\cC$ is defined to be $\rho\parenv{\cC}\triangleq m\log\binom{n}{w}-\log\parenv{\abs{\cC}}$, where $\log(\cdot)$ is the logarithm function with base $2$. We aim to construct $\sparenv{m,(n,w);t,e}$-composite codes with large cardinality, i.e. small redundancy. 

$\abs{\cC}$ or equivalently, with small $\rho\parenv{\cC}$.

Let $p$ be a prime number. For any sequence $\bm{x}=x_0\cdots x_{n-1}\in\cA_Q^n$ and nonnegative integer $\ell$, we define $s_{\ell}^{p}\parenv{\bm{x}}\triangleq\sum_{j=0}^{n-1}j^{\ell}x_j\pmod{p}$, which is called the $\ell$-VT-syndrome modulo $p$ of $\bm{x}$. For a sequence $\bm{x}\in\cA_2^n$ with a given Hamming weight, the values of $s_{1}^{p}\parenv{\bm{x}},\ldots,s_{e}^{p}\parenv{\bm{x}}$ can be leveraged to correct $e$ composite-asymmetric errors occurring in $\bm{x}$, as shown in next two lemmas.
\begin{lemma}\cite[the proof of Lemma 4.1]{Lara2010SIAM}\label{lem_pre1}
  Let $Q>1$ be an integer and $q\ge\max\{Q,t+1\}$ be a prime nymber. Suppose that $0\le\alpha_1,\ldots,\alpha_t<Q$ satisfy
  \begin{equation}\label{eq_pre1}
    \left\{
    \begin{array}{c}
     \alpha_1+\cdots+\alpha_t\equiv a_1\pmod{q}\\
     \alpha_1^2+\cdots+\alpha_t^2\equiv a_2\pmod{q}\\
     \vdots\\
     \alpha_1^t+\cdots+\alpha_t^t\equiv a_t\pmod{q}
    \end{array}
    \right.,
  \end{equation}
where $a_1,\ldots,a_t\in\mathbb{Z}_q$ are given. Then the multiset $\multiset{\alpha_1,\ldots,\alpha_t}$ can be uniquely determined by \Cref{eq_pre1}.
\end{lemma}

\begin{remark}\label{rmk_exactte}
  Our codewords are taken from $\Sigma_w^{m\times n}$ and thus each row of a codeword $\cX$ has Hamming weight $w$. This means that once a corrupted version of $\cX$ is received, we can know which rows are erroneous and how many asymmetric errors there are in each row.
  Due to this observation, when proving the correctness of a constructed code, it is sufficient to assume that there are exactly $t$ rows are suffered from errors and in each erroneous row, there are exactly $e$ composite-asymmetric errors.
\end{remark}

The following lemma follows directly from \Cref{lem_pre1}. For completeness, we present its proof in Appendix~\ref{appendix_pfpre2}.
\begin{lemma}\label{lem_pre2}
  Let $\bm{x}\in\Sigma_2^n$ be a sequence of Hamming weight at least $e$, where $e<n$. Let $p\ge n$ a prime. Suppose that $\bm{y}$ is obtained from $\bm{x}$ by exactly $e$ composite-asymmetric errors. Then given $s_{\ell}^p(\bm{x})$ for all $1\le\ell\le e$, $\bm{x}$ can be uniquely recovered from $\bm{y}$.
\end{lemma}

For integer $e\ge 1$ and a prime $p$, define
\begin{equation*}
  f_{p,e}\parenv{\bm{x}}\triangleq s_{1}^{p}\parenv{\bm{x}}+p\cdot s_{2}^{p}\parenv{\bm{x}}+\cdots p^{e-1}\cdot s_{e}^{p}\parenv{\bm{x}}.
\end{equation*}
Then $0\le f_{p,e}\parenv{\bm{x}}\le p^e-1$. In other words, $f_{p,e}\parenv{\bm{x}}$ is a symbol in the alphabet $\cA_{p^e}$.
\begin{fact}\label{fact_pre1}
 The tuple $\parenv{s_{1}^{p}\parenv{\bm{x}},\ldots,s_{e}^{p}\parenv{\bm{x}}}$ is uniquely determined by $ f_{p,e}\parenv{\bm{x}}$.
\end{fact}

For each $\cX\in\Sigma_{w}^{m\times n}$, define
$$
f\parenv{\cX}\triangleq\parenv{f_{p,e}\parenv{\cX_0},\ldots,f_{p,e}\parenv{\cX_{m-1}}}.
$$

\begin{proposition}\label{prop_main}
Let $p\ge n$ be a prime. Let $\cX\in\Sigma_w^{m\times n}$ and $\cY$ be obtained from $\cX$ by a $(t,e)$-CAE. According to \Cref{rmk_exactte}, we can assume that there are exactly $t$ erroneous rows $\cY_{i_1},\ldots,\cY_{i_t}$ and $\wth{\cY_{i_j}}=w-e$ for all $1\le j\le t$. By \Cref{lem_pre2} and \Cref{fact_pre1}, to decode $\cX$ from $\cY$, it is sufficient to obtain the values of $f_{p,e}(\cX_{i_j})$ for all $1\le j\le t$. In addition, since $\cY_i=\cX_i$ when $i\notin\{i_1,\ldots,i_t\}$, we have $f_{p,e}(\cX_i)=f_{p,e}(\cY_i)$ for all $i\notin\{i_1,\ldots,i_t\}$. Now the problem boils down to correcting erasures in $f(\cX)$ at positions $i_1,\ldots,i_t$.
\end{proposition}

In \cite{Omer2024ISIT}, the authors viewed the vector $\parenv{s_1^p(\bm{x}),\ldots,s_1^p(\bm{x})}$ as an element in the finite field $\mathbb{F}_{p^e}$. Here, we encode the vector $\parenv{s_1^p(\bm{x}),\ldots,s_1^p(\bm{x})}$ into $f_{p,e}(\bm{x})$, which is a symbol in the alphabet $\Sigma_{p^e}$. As an illustration of our method, next we show an alternative construction to \cite[Construction 1]{Omer2024ISIT}.

\begin{construction}\label{constuction_pre1}
  Let $m,n>1$, $1\le w<n$, $1\le e\le w$ and $1\le t<m$. Let $p$ be the smallest prime in $\sparenv{n,2n}$ and $q\ge p^e$ is a prime. Suppose that $\cC_t\subseteq\mathbb{F}_q^m$ is a coset of an $\sparenv{m,k,t+1}$ code over $\mathbb{F}_q$. The code $\cC$ is defined as
  $$
  \cC\triangleq\mathset{\cX\in\Sigma_{w}^{m\times n}:~f\parenv{\cX}\in\cC_t}.
  $$
\end{construction}

Then the following theorem follows from \Cref{prop_main}.
\begin{theorem}\label{thm_pre1}
  The code $\cC$ is a $(t,e)$-CAECC. When $m\le q<2m$, there is a choice of $\cC_t$ such that the redundancy is at most $t\log\parenv{q}\in\sparenv{t\log(m),t\log(m)+t}$.
\end{theorem}

Note that in \cite[Construction 1]{Omer2024ISIT}, the coset $\cC_t$ is over the field $\mathbb{F}_{p^e}$, whose size is a prime power. In our construction, the code $\cC_t$ is over a field of size $q$, which is a prime number. In our construction, all calculations can be done in the context of residue integers.
It should be noted that the redundancy of our code is larger than that of the code in \cite[Corollary 1]{Omer2024ISIT} by at most $t$. But our construction can allow larger values of $m$, i.e, $m>p^e$. If we replace $\mathbb{F}_{p^{e}}$ with $\mathbb{F}_{p^{el}}$ where $l>1$ in \cite[Construction 1]{Omer2024ISIT}, then \cite[Corollary 1]{Omer2024ISIT} allows larger values of $m$. However, the computations in a larger field of prime power size will be more complicated.

Apart from advantages mentioned above, encoding the vector $\parenv{s_1^p(\bm{x}),\ldots,s_1^p(\bm{x})}$ into a symbol in the alphabet $\Sigma_{p^e}$ enables us to construct new codes, which are defined according to constraints given by congruent equations. As a result, this approach leads to codes with more flexible parameters and in some cases codes with optimal levels of redundancy, which will be shown in the next section. Throughout this paper, whenever it comes to $\sparenv{m,(n,w);t,e}$-composite codes, it is always assumed that $m$ is large, and $n,w,e,t$ are constants.

\section{Uniquely decodable CAECCs}\label{sec_unique}
In this section, we first improve the sphere-packing upper bound on the size of an $\sparenv{m,(n,w);t,e}$-composite code given in \cite[Theorem 2]{Omer2024ISIT} to a Johnson type upper bound, when $t$ is odd. After that, we present several constructions of $\sparenv{m,(n,w);t,e}$-composite codes. Recall that it is always assumed that $m$ is large, while $n,w,e,t$ are constants. We aim to improve the cardinality or redundancy of codes in this regime of parameters. In \Cref{subsection_outer}, we construct codes which attain redundancy which is much smaller than those given in~\cite[Theorem 1]{Omer2024ISIT}. To further improve the code size, we introduce the notion of weak $B_e$ sets and use them to determine the error positions in each row. It turns out that the code size can be improved by a linear factor.

\subsection{Bounds}
For any two sequences $\bm{x},\bm{y}\in\Sigma_q^n$, the Hamming distance between them is denoted by $d_H\parenv{\bm{x},\bm{y}}$.
\begin{definition}\cite{Omer2024ISIT}\label{dfn_etdistance}
  For integer $e\ge0$, the $e$-Hamming distance of $\cX,\cY\in\{0,1\}^{m\times n}$, $\ed{e}{\cX,\cY}$, is defined as
  $$
  \ed{e}{\cX,\cY}=
  \begin{cases}
    \infty, & \mbox{if } d_H\parenv{\cX_i,\cY_i}>e\text{ for some }i, \\
    \abs{\mathset{i:~\cX_i\ne\cY_i}}, & \mbox{otherwise}.
  \end{cases}
  $$
  The $e$-Hamming distance of a code $\cC$ is defined to be $\ed{e}{\cC}\triangleq\min\mathset{\ed{e}{\cX,\cY}:~\cX,\cY\in\cC,\cX\ne\cY}$.
\end{definition}

It is proved in \cite[Lemma 1]{Omer2024ISIT} that $\cC\subseteq\Sigma_w^{m\times n}$ is a $(t,e)$-CAECC if and only if $\ed{2e}{\cC}\ge t+1$. Although we call the function $\ed{e}{\cdot,\cdot}$ ``distance", it does not satisfy the triangle inequality, as shown in next example.
\begin{example}
  Let
  \begin{equation*}
    \cX=
    \begin{pmatrix}
      1 & 1 & 1 & 0 & 0 \\
      0 & 1 & 1 & 1 & 0 \\
      0 & 0 & 1 & 1 & 1
    \end{pmatrix},
     \cY=
    \begin{pmatrix}
      0 & 1 & 1 & 1 & 0 \\
      1 & 1 & 1 & 0 & 0 \\
      0 & 0 & 1 & 1 & 1
    \end{pmatrix},
     \cZ=
    \begin{pmatrix}
      0 & 1 & 1 & 1 & 0 \\
      1 & 1 & 0 & 0 & 1 \\
      0 & 1 & 1 & 1 & 0
    \end{pmatrix}.
  \end{equation*}
  It is easy to see that $\ed{2}{\cX,\cY}=\ed{2}{\cY,\cZ}=2$. On the other hand, since $d_H\parenv{\cX_2,\cZ_2}=4$, according to \Cref{dfn_etdistance}, we have $\ed{2}{\cX,\cZ}=\infty$ and hence $\ed{2}{\cX,\cZ}>\ed{2}{\cX,\cY}+\ed{2}{\cY,\cZ}$.
\end{example}

Fortunately, we have the following lemma.
\begin{lemma}\label{lem_trinequality}
  For any $\cX,\cY,\cZ\in\Sigma_w^{m\times n}$ and integer $e\ge0$, it holds that
  $$
  \ed{e}{\cX,\cZ}\le\ed{\floorenv{e/2}}{\cX,\cY}+\ed{\floorenv{e/2}}{\cY,\cZ}.
  $$
\end{lemma}
\begin{IEEEproof}
For any two sequences $\bm{x}$ and $\bm{y}$ with the same Hamming weight, it is clear that $d_H\parenv{\bm{x},\bm{y}}$ is an even integer. Thus we may assume that $e$ is even. When $\ed{\floorenv{e/2}}{\cX,\cY}=\infty$ or $\ed{\floorenv{e/2}}{\cY,\cZ}=\infty$, the conclusion is trivial. Next, assume $\ed{\floorenv{e/2}}{\cX,\cY}<\infty$ and $\ed{\floorenv{e/2}}{\cY,\cZ}<\infty$. Then if $e=0$ or $2$, we have $\cX=\cY=\cZ$ and then the conclusion follows.

Now assume $e\ge4$, $\ed{\floorenv{e/2}}{\cX,\cY}=d_1$ and $\ed{\floorenv{e/2}}{\cY,\cZ}=d_2$. According to \Cref{dfn_etdistance}, there exist $I_1=\mathset{i_1,\ldots,i_{r_1}}$, $I_2=\mathset{j_1,\ldots,j_{r_2}}$, $I_3=\mathset{k_1,\ldots,k_{r_3}}\subseteq[m]$ such that $\cX_i\ne\cY_i$ if and only if $i\in I_1\cup I_2$, and $\cY_i\ne\cZ_i$ if and only if $i\in I_2\cup I_3$. Here, we have $r_1,r_2,r_3\ge0$, $r_1+r_2=d_1$, $r_2+r_3=d_2$ and $I_1,I_2,I_3$ are mutually disjoint. Therefore, it holds that $\cX_i\ne\cZ_i$ only if $i\in I_1\cup I_2\cup I_3$ and that $j \in I_2$ if and only if $\cX_j \neq \cY_j$ and $\cY_j \neq \cZ_j$.

For $i\in I_1$, since $d_H\parenv{\cX_i,\cY_i}\le\floorenv{e/2}$ and $d_H\parenv{\cY_i,\cZ_i}=0$, we have $d_H\parenv{\cX_i,\cZ_i}\le\floorenv{e/2}$. Similarly, it holds that $d_H\parenv{\cX_i,\cZ_i}\le\floorenv{e/2}$ for $i\in I_3$. When $i\in I_2$, since $d_H\parenv{\cX_i,\cY_i}\le\floorenv{e/2}$ and $d_H\parenv{\cY_i,\cZ_i}\le\floorenv{e/2}$, we have $d_H\parenv{\cX_i,\cZ_i}\le e$. In other words, for each $i\in I_1\cup I_2\cup I_3$, it holds that $d_H\parenv{\cX_i,\cZ_i}\le e$. Therefore, $\ed{e}{\cX,\cZ}\le\abs{I_1\cup I_2\cup I_3}=d_1+d_2-r_2\le\ed{\floorenv{e/2}}{\cX,\cY}+\ed{\floorenv{e/2}}{\cY,\cZ}$. Now the proof is completed.
\end{IEEEproof}

For $\cX\in\Sigma_w^{m\times n}$, $0\le e\le 2w$ and $0\le r\le m$, define
$$
\begin{array}{l}
\eS{0}{\cX,r}=\mathset{\cX},\\
\eS{e}{\cX,r}=\mathset{\cY\in\Sigma_w^{m\times n}:~\ed{e}{\cX,\cY}=r},\text{ for }e>0
\end{array}
$$
and
\begin{equation}\label{eq_dfnball1}
\eB{e}{\cX,r}=\mathop{\cup}\limits_{s=0}^{r}\eS{e}{\cX,s}.
\end{equation}


\begin{corollary}\label{cor_packingradius}
If $\cC$ is a $(t,e)$-CAECC, then $\eB{e}{\cX,\floorenv{t/2}}\cap\eB{e}{\cY,\floorenv{t/2}}=\emptyset$ for any two distinct $\cX$, $\cY\in\cC$.
\end{corollary}
\begin{IEEEproof}
Suppose on the contrary that $\eB{e}{\cX,\floorenv{t/2}}\cap\eB{e}{\cY,\floorenv{t/2}}\ne\emptyset$. Let $\cZ\in\eB{e}{\cX,\floorenv{t/2}}\cap\eB{e}{\cY,\floorenv{t/2}}$. Then by \Cref{lem_trinequality}, $\ed{2e}{\cX,\cY}\le\ed{e}{\cX,\cZ}+\ed{e}{\cY,\cZ}\le t$, which contradicts the fact that $\cC$ is a $(t,e)$-CAECC.
\end{IEEEproof}

If $r\le m$, by definition, it is easy to see that $\abs{\eS{e}{\cX,r}}=\abs{\eS{(e-1)}{\cX,r}}$ when $e$ is odd and
\begin{equation}\label{eq_spheresize1}
\abs{\eS{e}{\cX,r}}=\binom{m}{r}\sparenv{\sum_{l=1}^{\frac{e}{2}}\binom{w}{l}\binom{n-w}{l}}^r
\end{equation}
when $e\ge2$ is even. Since $\abs{\eS{e}{\cX,r}}$ is independent of $\cX$, we denote the size of any $\eS{e}{\cX,r}$ by $\vS{e}{r}$. Similarly, we denote the size of any $\eB{e}{\cX,r}$ by $\vB{e}{r}$.

\begin{theorem}\label{thm_uniquelybound1}
  Let $0<e\le w$ and $1\le t< m$. Suppose that $\cC\subseteq\Sigma_w^{m\times n}$ is a $(t,e)$-CAECC.
  \begin{itemize}
    \item(\textbf{Sphere-packing bound\cite[Theorem 2]{Omer2024ISIT}}) For all $t$, we have
    $$
    \abs{\cC}\le\frac{\binom{n}{w}^m}{\sum_{s=0}^{\floorenv{t/2}}\binom{m}{s}\sparenv{\sum_{l=1}^{\floorenv{e/2}}\binom{w}{l}\binom{n-w}{l}}^s}.
    $$
    \item(\textbf{Johnson bound}) If $t$ is odd, then
    $$
    \abs{\cC}\le\frac{\binom{n}{w}^m}{\sum_{s=0}^{(t-1)/2}\binom{m}{s}\sparenv{\sum_{l=1}^{\floorenv{e/2}}\binom{w}{l}\binom{n-w}{l}}^s+\frac{\binom{m}{(t+1)/2}\sparenv{\sum_{l=1}^{\floorenv{e/2}}\binom{w}{l}\binom{n-w}{l}}^{(t+1)/2}}{\floorenv{\frac{m}{(t+1)/2}}}}.
    $$
  \end{itemize}
\end{theorem}


\begin{IEEEproof}
Our main goal is to prove the Johnson bound and the sphere-packing bound is just a byproduct.
For $\cY\in\Sigma_w^{m\times n}$, define $\ed{e}{\cY,\cC}=\min\mathset{\ed{e}{\cY,\cX}:\cX\in\cC}$. Let $\cN=\mathset{\cY\in\Sigma_w^{m\times n}:\ed{e}{\cY,\cC}=\floorenv{\frac{t}{2}}+1}$. Then $\cN\cap\eB{e}{\cX,\floorenv{t/2}}=\emptyset$ for all $\cX\in\cC$. By \Cref{cor_packingradius}, we have
\begin{equation}\label{eq_uniqueJohnsonbound1}
  \abs{\cC}\vB{e}{\floorenv{t/2}}+\abs{\cN}\le\binom{n}{w}^m.
\end{equation}
In order to upper bound $\abs{\cC}$, it is sufficient to derive a lower bound on $\abs{\cN}$. 
The statement regarding the sphere-packing in the theorem follows from \Cref{eq_dfnball1,eq_spheresize1,eq_uniqueJohnsonbound1} by observing that $|\cN| \geq 0$. In the following, we follow similar ideas in the proof of \cite[Theorem 2.3.8]{huffman_pless_2003} to improve the sphere-packing bound by identifying a tighter lower bound on $|\cN|$ when $t$ is odd. 
Let
$$
\cP=\mathset{\parenv{\cX,\cY}\in\cC\times\cN:~\ed{e}{\cX,\cY}=\frac{t+1}{2}}.
$$

For each $\cX\in\cC$, there are exactly $\vS{e}{\frac{t+1}{2}}$ $\cY$'s such that $\ed{e}{\cY,\cX}=(t+1)/2$. Suppose $\ed{e}{\cY,\cX}=\frac{t+1}{2}$, where $\cX\in\cC$. Then for any $\cX^\prime\in\cC$ and $\cX^\prime\ne\cX$, we have $t+1\le\ed{2e}{\cX,\cX^\prime}\le\ed{e}{\cX,\cY}+\ed{e}{\cY,\cX^\prime}$ by \Cref{lem_trinequality}, which implies $\ed{e}{\cY,\cX^\prime}\ge\frac{t+1}{2}$ and so $\cY\in\cN$. Therefore, we conclude that
\begin{equation}\label{eq_uniqueJohnsonbound2}
  \abs{\cP}=\abs{\cC}\vS{e}{\frac{t+1}{2}}.
\end{equation}

For each $\cY\in\cN$, suppose $\cX,\cX^\prime\in\cC$ such that $\cX\ne\cX^\prime$ and $\ed{e}{\cX,\cY}=\ed{e}{\cX^\prime,\cY}=\frac{t+1}{2}$. Then we have $t+1\le\ed{2e}{\cX,\cX^\prime}\le\ed{e}{\cX,\cY}+\ed{e}{\cX^\prime,\cY}=t+1$ and hence $\ed{2e}{\cX,\cX^\prime}=t+1$. By \Cref{dfn_etdistance}, there are two subsets $I_1,I_2\subseteq[m]$ with $\abs{I_1}=\abs{I_2}=\frac{t+1}{2}$, such that $\cX_i\ne\cY_i$ if and only if $i\in I_1$, and $\cX_i^\prime\ne\cY_i$ if and only if $i\in I_2$. This implies that $\cX_i\ne\cX_i^\prime$ only if $i\in I_1\cup I_2$. Then it follows that $t+1=\ed{2e}{\cX,\cX^\prime}\le\abs{I_1}+\abs{I_2}-\abs{I_1\cap I_2}=t+1-\abs{I_1\cap I_2}$, which implies $I_1\cap I_2=\emptyset$. Therefore, for each $\cY\in\cN$, there are at most $\floorenv{\frac{m}{(t+1)/2}}$ codewords $\cX\in\cC$ such that $\parenv{\cX,\cY}\in\cP$. Hence,
\begin{equation}\label{eq_uniqueJohnsonbound3}
 \abs{\cP}\le\abs{\cN}\floorenv{\frac{m}{(t+1)/2}}.
\end{equation}
Lastly, the second conclusion follows from \Cref{eq_dfnball1,eq_spheresize1,eq_uniqueJohnsonbound1,eq_uniqueJohnsonbound2,eq_uniqueJohnsonbound3}.
\end{IEEEproof}

Next we show by how much the Johnson bound improves the sphere-packing bound. Suppose that $t\ge3$ is odd. Let $\sigma_1=\sum_{s=0}^{\floorenv{t/2}}\binom{m}{s}\sparenv{\sum_{l=1}^{\floorenv{e/2}}\binom{w}{l}\binom{n-w}{l}}^s$ and  $\sigma_2=\frac{\binom{m}{(t+1)/2}\sparenv{\sum_{l=1}^{\floorenv{e/2}}\binom{w}{l}\binom{n-w}{l}}^{(t+1)/2}}{\floorenv{\frac{m}{(t+1)/2}}}$. Then, the Johnson bound improves the sphere packing bound by a factor $(\sigma_1+\sigma_2)/\sigma_1$. Since $\sigma_1\le\frac{t+1}{2}\binom{m}{(t-1)/2}\sparenv{\sum_{l=1}^{\floorenv{e/2}}\binom{w}{l}\binom{n-w}{l}}^{\frac{t-1}{2}}$, it follows that
\begin{align*}
\frac{\sigma_1+\sigma_2}{\sigma_1}=1+\frac{\sigma_2}{\sigma_1}&\ge 1+\frac{\binom{m}{(t+1)/2}\sparenv{\sum_{l=1}^{\floorenv{e/2}}\binom{w}{l}\binom{n-w}{l}}^{(t+1)/2}}{m\binom{m}{(t-1)/2}\sparenv{\sum_{l=1}^{\floorenv{e/2}}\binom{w}{l}\binom{n-w}{l}}^{(t-1)/2}}\\
&=1+\frac{t-1}{t+1}\parenv{1-\frac{t-1}{2m}}\sum_{l=1}^{\floorenv{e/2}}\binom{w}{l}\binom{n-w}{l}\\
&\ge 1+\frac{1}{2}\cdot\frac{t-1}{t+1}\sum_{l=1}^{\floorenv{e/2}}\binom{w}{l}\binom{n-w}{l}.
\end{align*}
For example, when $t=3$ and $e=2$, we have $(\sigma_1+\sigma_2)/\sigma_1\ge 1+w(n-w)/4\ge (n+3)/4$.

\subsection{Constructions---improving the outer code}\label{subsection_outer}
When $t,n,w,e$ are assumed to be constants while $m$ is a variable, \Cref{thm_uniquelybound1} implies $\rho\parenv{\cC}\ge\floorenv{t/2}\log(m)+\Omega(1)$ for any $[m,(n,w);t,e]$-composite code $\cC$. The codes in both \Cref{thm_pre1} and \cite[Corollary 1]{Omer2024ISIT} have redundancy $t\log(m)+O(1)$. In this section, we give a construction which reduces the redundancy to $(t-1)\log(m)+O(1)$ for all $t\ge 2$. This implies that when $t=2$, there are codes which are optimal up to a constant in terms of redundancy. Furthermore, when $m$ is a prime power (this prime power depends on $n$ and $e$) and $t\ge3$, we further improve the redundancy to $\parenv{t-1-k}\log(m)+O(1)$, where $k=\floorenv{\frac{t-1}{q}}$ or $\frac{1}{t-1}$.

Before proceeding, we briefly explain our idea. Recall that in \Cref{constuction_pre1}, it is required that the vector $f\parenv{\cX}$ belongs to a coset $\cC_t$ of an $[m,k,\ge t+1]$ linear code. Let $\mathbf{H}$ be a parity-check matrix of this linear code and $\cC_t=\mathset{\bm{c}\in\mathbb{F}_q^m:\mathbf{H}\cdot \bm{c}^T=\bm{b}^T}$, where $\bm{b}\in\mathbb{F}_q^{m-k}$. Then $f\parenv{\cX}\in\cC_t$ is equivalent to
\begin{equation}\label{eq_idea1}
  \mathbf{H}\cdot f\parenv{\cX}^T=\bm{b}^T.
\end{equation}
In other words, the vector $f\parenv{\cX}$ should satisfy $m-k$ equations over the field $\mathbb{F}_q$ given in (\ref{eq_idea1}). Therefore, the redundancy can only be guaranteed to be at most $(m-k)\log(q)$ (note that $m$ is a function of $q$). Since $\mathbf{H}$ is a parity-check matrix of a linear code, we can assume that the first row of $\mathbf{H}$ is $\parenv{1,1,\ldots,1}$. In other words, it is a vector consisting of $1$'s. To further reduce the redundancy when $m$ is large, we can replace the first equation in (\ref{eq_idea1}) by $\sum_{r=0}^{m-1}f_{p,e}\parenv{\cX_r}\equiv a\pmod{N}$, where $N$ is independent of $m$ and $a\in\mathbb{Z}_N$ is given. Then,  (\ref{eq_idea1}) can be replaced by
\begin{equation*}
  \begin{aligned}
  \sum_{r=0}^{m-1}f_{p,e}\parenv{\cX_r}\equiv a\pmod{N},\\
  \mathbf{H}^\prime\cdot f\parenv{\cX}^T=\parenv{\bm{b}^\prime}^T,
  \end{aligned}
\end{equation*}
where $\mathbf{H}^\prime$ is a matrix obtained from $\mathbf{H}$ by deleting the first row and $\bm{b}^\prime\in\mathbb{F}_q^{m-k-1}$. As will be clear later, the number $N$ only depends on $n$, $e$ and $t$, which means it can be treated as a constant. This implies that the redundancy of the resulting construction is at most
$(m-k-1)\log(q)+\log(N)=(m-k-1)\log(q)+O(1)$.

\bigskip
\subsubsection{Constructions for general $t$}
\begin{construction}[$t\ge2$]\label{constuction_uniquegeneral}
  Let $m,n>1$, $1\le w<n$, $1\le e\le w$ and $2\le t<m$. Let $p$ be the smallest prime in $\sparenv{n,2n}$ and $q$ be a prime and $q\ge\max\mathset{m,p^e}$.
  For any $\bm{a}=\parenv{a_0,a_1,a_2,\ldots,a_{t-1}}\in\mathbb{Z}_{tp^e}\times\mathbb{Z}_{q}^{t-1}$, the code $\cU$ is defined as
  $$
  \cU\triangleq\mathset{\cX\in\Sigma_{w}^{m\times n}:~\sum_{r=0}^{m-1}f_{p,e}\parenv{\cX_{r}}\equiv a_0\pmod{tp^e},\sum_{r=0}^{m-1}r^{\ell}\cdot f_{p,e}\parenv{\cX_r}\equiv a_{\ell}\pmod{q},\forall\ell\in\sparenv{1,t-1}}.
  $$
\end{construction}

\begin{theorem}\label{thm_uniquegeneral}
  The code $\cU$ is a $(t,e)$-CAECC. When $m\ge p^e$, there is an $\bm{a}$ such that the redundancy is at most $(t-1)\log(q)+e\log(p)+\log(t)\in\sparenv{(t-1)\log(m)+e\log(n)+\log(t),(t-1)\log(m)+e\log(n)+\log(t)+t-1+e}$. In particular, when $t=2$, the redundancy is optimal up to a constant.
\end{theorem}
\begin{IEEEproof}
Let $\cY$ be obtained from a codeword $\cX$ by a $(t,e)$-composite asymmetric error. Comparing $\cY$ with $\cX$, we can know which rows suffered from errors. By \Cref{rmk_exactte}, suppose that there are $t$ erroneous rows $\cX_{i_1},\ldots,\cX_{i_t}$. Let $a_0^\prime=\parenv{a_0-\sum_{r\ne i_1,\ldots,i_t}f_{p,e}\parenv{\cX_{r}}\pmod{tp^e}}$ and $a_{\ell}^\prime=\parenv{a_{\ell}-\sum_{r\ne i_1,\ldots,i_t}r^{\ell}\cdot f_{p,e}\parenv{\cX_{r}}\pmod{q}}$ for all $1\le\ell\le t-1$. Then
\begin{equation}\label{eq_uniquegeneral1}
\left\{
  \begin{array}{c}
    f_{p,e}(\cX_{i_1})+\cdots + f_{p,e}(\cX_{i_t})=a_0^\prime\\
    i_1\cdot f_{p,e}(\cX_{i_1})+\cdots +i_t\cdot f_{p,e}(\cX_{i_t})\equiv a_1^\prime\pmod{q}\\
    \vdots\\
    i_1^{t-1}\cdot f_{p,e}(\cX_{i_1})+\cdots +i_t^{t-1}\cdot f_{p,e}(\cX_{i_t})\equiv a_{t-1}^\prime\pmod{q}
  \end{array}
  \right..
\end{equation}
The first line in \Cref{eq_uniquegeneral1} follows from the fact that $0\le a_0^\prime<tp^e$ and $0\le \parenv{a_0-\sum_{r\ne i_1,\ldots,i_t}f_{p,e}\parenv{\cX_{r}}\pmod{tp^e}}<tp^e$. By \Cref{eq_uniquegeneral1}, we have
 \begin{equation}\label{eq_uniquegeneral2}
\left\{
  \begin{array}{c}
    f_{p,e}(\cX_{i_1})+\cdots + f_{p,e}(\cX_{i_t})\equiv a_0^\prime\pmod{q}\\
    i_1\cdot f_{p,e}(\cX_{i_1})+\cdots +i_t\cdot f_{p,e}(\cX_{i_s})\equiv a_1^\prime\pmod{q}\\
    \vdots\\
    i_1^{t-1}\cdot f_{p,e}(\cX_{i_1})+\cdots +i_t^{t-1}\cdot f_{p,e}(\cX_{i_t})\equiv a_{t-1}^\prime\pmod{q}
  \end{array}
  \right..
\end{equation}
The coefficient matrix of system (\ref{eq_uniquegeneral2}) is
$$
A=
\begin{pmatrix}
  1 & 1 & \cdots & 1 \\
  i_1 & i_2 & \cdots & i_t \\
  \vdots & \vdots & \cdots & \vdots \\
  i_1^{t-1} & i_2^{t-1} & \cdots & i_t^{t-1}
\end{pmatrix}.
$$
Since $q\ge m$, we have $i_{j}\not\equiv i_{k}\pmod{q}$ for all $1\le j\ne k\le t$. Hence, the matrix $A$ is a Vandermonde matrix in the finite field $\mathbb{F}_q$. Therefore, system (\ref{eq_uniquegeneral2}) has a unique solution in $\mathbb{F}_q$: $\parenv{f_{p,e}(\cX_{i_1})\pmod{q},\ldots,f_{p,e}(\cX_{i_t})\pmod{q}}$. Since $q\ge p^e$, we conclude that $f_{p,e}(\cX_{i_j})=f_{p,e}(\cX_{i_j})\pmod{q}$ for all $1\le j\le t$. Now according to \Cref{prop_main}, the proof is completed.
\end{IEEEproof}

\begin{remark}
Comparing with \cite[Theorem 2]{Omer2024ISIT}, when $t=2$, the redundancy of $\cC$ is larger than the lower bound by at most $e\log\frac{n}{w(n-w)}+2\parenv{e\log(e)+1}+e$. When $e$ is a constant and $w(n-w)=\Theta\parenv{n}$, $\cC$ is optimal up to a constant in terms of redundancy.
\end{remark}

Next, we give another construction of optimal $(2,e)$-CAECCs when $\gcd\parenv{m,(p^e-1)!}=1$. This construction will improve the redundancy by roughly one bit.
\begin{lemma}\label{lem_uniquetis2}
Let $Q>1,m>1$ be two positive integers, $0\le\delta_1,\delta_2<Q$  and $0\le i<j< m$. Suppose that $i,j$ are known.
If $\gcd\parenv{m,(Q-1)!}=1$, given $\parenv{i\delta_1+j\delta_2}\pmod{m}$ and $\delta_1+\delta_2$, the values of $\delta_1$ and $\delta_2$ can be recovered.
\end{lemma}
\begin{IEEEproof}
Suppose that there are $0\le\sigma_1,\sigma_2<Q$ such that
\begin{equation*}
  \left\{
  \begin{array}{l}
    i\delta_1+j\delta_2\equiv i\sigma_1+j\sigma_2 \pmod{m}, \\
    \delta_1+\delta_2=\sigma_1+\sigma_2.
  \end{array}
  \right.
\end{equation*}
Then we have
\begin{equation}\label{eq_uniqueis2}
  (i-j)(\delta_1-\sigma_1)\equiv0\pmod{m}.
\end{equation}

Suppose $\delta_1\ne\sigma_1$. Then we have $-(Q-1)\le\delta_1-\sigma_1<0$ or $0<\delta_1-\sigma_1\le Q-1$. Since $\gcd\parenv{m,(Q-1)!}=1$, it holds that $\gcd(m,\delta_1-\sigma_1)=1$. Now it follows from \Cref{eq_uniqueis2} that $i\equiv j\mod{m}$, which contradicts the fact that $0\le i<j<m$. Therefore, we can conclude that $\delta_1=\sigma_1$ and $\delta_2=\sigma_2$.
\end{IEEEproof}

\begin{construction}[$t=2$]\label{constuction_uniquedecoding2}
Let $m,n>1$, $1\le w<n$, $1\le e\le w$. Let $p$ be the smallest prime in $\sparenv{n,2n}$.
Suppose $\gcd\parenv{m,(p^e-1)!}=1$. For any $\bm{a}=\parenv{a_1,a_2}\in\mathbb{Z}_{2p^e}\times\mathbb{Z}_{m}$, the code $\cU^\prime$ is defined as
  $$
  \cU^\prime\triangleq\mathset{\cX\in\Sigma_{w}^{m\times n}:\sum_{r=0}^{m-1}f_{p,e}\parenv{\cX_r}\equiv a_1\pmod{2p^e},\sum_{r=0}^{m-1}r\cdot f_{p,e}\parenv{\cX_{r}}\equiv a_2\pmod{m}}.
  $$
\end{construction}

\begin{theorem}\label{thm_unique2}
  The code $\cU^\prime$ is a $(2,e)$-CAECC. And there is some $\bm{a}$ such that the redundancy of $\cU_{2}^\prime$ is at most $\log(m)+e\log(p)+1\in\sparenv{\log(m)+e\log(n)+1, \log(m)+e\log(n)+e+1}$.
\end{theorem}
\begin{IEEEproof}
Let $\cY$ be obtained from a codeword $\cX$ by a $(2,e)$-CAE. Comparing $\cY$ with $\cX$, we can determine which rows suffered errors. Suppose that $\cX_i,\cX_j$ suffered errors, where $0\le i<j<m$. Let $Q=p^e$, $\delta_1=f_{p,e}\parenv{\cX_i}$ and $\delta_2=f_{p,e}\parenv{\cX_j}$. By computing  $\parenv{a_1-\sum_{r\ne i,j}f_{p,e}\parenv{\cX_{r}}}\pmod{2p^e}$ and $\parenv{a_2-\sum_{r\ne i,j}r\cdot f_{p,e}\parenv{\cX_{r}}}\pmod{m}$, we can find the values of $\parenv{i\delta_1+j\delta_2}\pmod{m}$ and $\delta_1+\delta_2$, respectively. Then, according to \Cref{lem_uniquetis2}, the values of $\delta_1$ and $\delta_2$, and hence values of $f_{p,e}\parenv{\cX_i}$ and $f_{p,e}\parenv{\cX_j}$, can be obtained. By \Cref{prop_main}, the original codeword can be uniquely recovered.
\end{IEEEproof}

Suppose $t=2$ and $\gcd\parenv{m,(p^e-1)!}=1$. The code sizes guaranteed by \Cref{thm_unique2} and \Cref{thm_uniquegeneral} are $\binom{n}{w}^m/(2p^em)$ and $\binom{n}{w}^m/(2p^eq)$, respectively. It is easy to see that \Cref{thm_unique2} improves the code size by a factor $q/m$, which is greater than $1$ when $m$ is not a prime.

\bigskip
\subsubsection{Constructions when $m$ is a Prime Power}\label{sec_primepower}
In \Cref{thm_uniquegeneral}, it was shown that there are $(t,e)$-CAECCs with redundancy at most $(t-1)\log(m)+O(1)$, which implies there are $(2,e)$-CAECCs with redundancy optimal up to a constant. Let $p$ be as in \Cref{thm_uniquegeneral} and $q$ is the smallest prime in $[p^e,2p^e]$. In this subsection, we will show that when $m=q^l$ for some $l\ge2$ and $t\ge 3$, the redundancy can be further reduced to $(t-1-k)\log(m)+O(1)$, where $k=\floorenv{\frac{t-1}{q}}$ or $\frac{1}{t-1}$.

\begin{construction}
Let $n>1$, $1\le w<n$, $1\le e\le w$. Let $p$ be the smallest prime in $\sparenv{n,2n}$ and $q$ be the smallest prime in $[p^e,2p^e]$. Let $m=q^l$ for some $l\ge2$ and $2\le t<m$. Choose a narrow-sense BCH code $\cC_{BCH}$ of length $q^l-1$ and designed distance $t$ over $\mathbb{F}_q$. Let $\mathbf{H}$ be a parity-check matrix of $\cC_{BCH}$. Suppose the number of rows in $\mathbf{H}$ is $\rho$. Let $\widetilde{\mathbf{H}}=\parenv{\mathbf{H},\mathbf{0}}$, i.e., appending a column of $0$s to the right of $\mathbf{H}$.

For any $a_0\in\mathbb{Z}_{tp^e}$ and $\bm{b}=\parenv{b_1,\ldots,b_{\rho}}\in\mathbb{F}_q^{\rho}$, the code $\cU_{1}$ is defined as
  $$
  \cU_{1}\triangleq\mathset{\cX\in\Sigma_{w}^{m\times n}:~\sum_{r=0}^{m-1}f_{p,e}\parenv{\cX_r}\equiv a_0\pmod{tp^e},\widetilde{\mathbf{H}}\cdot f\parenv{\cX}^T=\bm{b}^T}.
  $$
  Recall that $f\parenv{\cX}=\parenv{f_{p,e}\parenv{\cX_0},\ldots,f_{p,e}\parenv{\cX_{m-1}}}$.
\end{construction}

\begin{theorem}\label{thm_tlarge}
  The code $\cU_1$ is a $(t,e)$-CAECC. There exist some $a_0$ and $\bm{b}$ such that $\rho\parenv{\cU_1}\le\frac{\rho}{l}\log\parenv{m}+\log\parenv{tp^e}\le\parenv{t-1-\floorenv{\frac{t-1}{q}}}\log(m)+e\log(p)+\log(t)$.
  In particular, if $t\ge kq+1$ for some $k\ge 1$, then $\rho\parenv{\cU_1}\le\parenv{t-1-k}\log(m)+e\log(p)+\log(t)$.
\end{theorem}
\begin{IEEEproof}
Let $\cY$ be obtained from a codeword $\cX$ by a $(t,e)$-composite asymmetric error. Comparing $\cY$ with $\cX$, we can know which rows suffered from errors. Suppose that the erroneous rows are $\cX_{i_1},\ldots,\cX_{i_t}$. Since $0\le\sum_{j=1}^{t}f_{p,e}\parenv{\cX_{i_j}}<tp^e$, we can know the value of $\sum_{j=1}^{t}f_{p,e}\parenv{\cX_{i_j}}$ by calculating $\parenv{a_0-\sum_{r\ne i_1,\ldots,i_t}f_{p,e}\parenv{\cX_{r}}}\pmod{tp^e}$.  Denote $b_0=\sum_{r=0}^{m-1}f_{p,e}\parenv{\cX_r}\pmod{q}$. Let $$
\widehat{\mathbf{H}}=
\begin{pmatrix}
  1\cdots 1 \\
  \text{------} \\
  \widetilde{\mathbf{H}},
\end{pmatrix}
$$
i.e., adding a row of $1$s on top of $\widetilde{\mathbf{H}}$. Then we have
\begin{equation}\label{eq_extendedbch}
\widehat{\mathbf{H}}\cdot
f\parenv{\cX}^T
=
\begin{pmatrix}
  b_0 \\
  \bm{b}^T
\end{pmatrix}.
\end{equation}
Let $\widehat{\cC}_{BCH}$ be the extended code of $\cC_{BCH}$ (see \cite[page 38]{LintGTM86} for the definition of the extended code of a given code). Then, \Cref{eq_extendedbch} implies that the vector $f\parenv{\cX}$ is in the coset of $\widehat{\cC}_{BCH}$ with syndrome $\parenv{b_0,\bm{b}}$. Since $\cC_{BCH}$ has minimum Hamming distance at least $t$, its extended code $\widehat{\cC}_{BCH}$ (and each of its coset) has minimum Hamming distance at least $t+1$. Therefore, given $f\parenv{\cY}$ and $\parenv{b_0,\bm{b}}$, we can uniquely recover $f\parenv{\cX}$ and hence the codeword $\cX$.

By the pigeonhole principle, there exists a $a_0$ and $\bm{b}$, such that $\rho\parenv{\cU_3}\le\rho\log(q)+\log\parenv{tp^e}=\frac{\rho}{l}\log(m)+\log\parenv{tp^e}$.
Since the BCH code $\cC_{BCH}$ is narrow-sense and with designed distance $t$, we have
$$
\rho\le l\parenv{t-1-\floorenv{\frac{t-1}{q}}}.
$$
Now the proof is completed.
\end{IEEEproof}

Note that the code $\cU_1$ 
requires less redundancy than $\cU$ only when $t\ge q+1$. Next, we will give a construction which  
requires less redundancy than $\cU$
 when $3\le t\le q$ and $m=q^l$ where $l>(t-2)!$. Let $q\ge t$ be a prime and $l>(t-2)!$. A $q$-ary $\sparenv{m=q^l,m-\rho-1,\ge t+1}$ linear code with the following parity-check matrix
$$
\widehat{\mathbf{H}}=
\begin{pmatrix}
  1\cdots 1 & 1 \\
  \mathbf{H} & \mathbf{0}
\end{pmatrix}
$$
was constructed in\cite[Theorem 5]{Yekhanin2004IT}, where $\mathbf{H}$ is a $\rho\times(q^l-1)$ matrix over $\mathbb{F}_q$ and $\rho\le(t-2)l+\ceilenv{\frac{l}{t-1}}$.

\begin{construction}[$3\le t\le q$]
Let $n>1$, $1\le w<n$, $1\le e\le w$. Let $p$ be the smallest prime in $\sparenv{n,2n}$ and $q$ be the smallest prime in $[p^e,2p^e]$. Let $3\le t\le q$ and $m=q^l$ for some $l>(t-2)!$. Let $\mathbf{H}$ be as above and $\widetilde{\mathbf{H}}=\parenv{\mathbf{H},\mathbf{0}}$.

For any $a_0\in\mathbb{Z}_{tp^e}$ and $\bm{b}=\parenv{b_1,\ldots,b_{\rho}}\in\mathbb{F}_q^{\rho}$, the code $\cU_{2}$ is defined as
  $$
  \cU_{2}\triangleq\mathset{\cX\in\Sigma_{w}^{m\times n}:~\sum_{r=0}^{m-1}f_{p,e}\parenv{\cX_r}\equiv a_0\pmod{tp^e},\widetilde{\mathbf{H}}\cdot f\parenv{\cX}^T=\bm{b}^T}.
  $$
\end{construction}

The following theorem can be proved in a similar way as \Cref{thm_tlarge}.
\begin{theorem}
  The code $\cU_2$ is a $(t,e)$-CAECC and there exist some $a_0$ and $\bm{b}$ such that $\rho\parenv{\cU_2}\le\frac{\rho}{l}\log\parenv{m}+\log\parenv{tp^e}\le\parenv{t-2+\frac{1}{t-1}}\log(m)+\log(q)+\log(tp^e)$.
\end{theorem}

\subsection{Constructions---improved encoding of each row}

Let $p$ be the smallest prime between $n$ and $2n$. All constructions above and in \cite{Omer2024ISIT} share one commonality: using $e$ constraints $s_{\ell}^p\parenv{\cX_i}$ for $1\le\ell\le e$, to identify the erroneous positions in $\cX_i$ and encode each row $\cX_i$ into a symbol in an alphabet of size $p^e$.
In this subsection, we will consider two alternative methods of encoding the information contained within each row of our codewords. In the first method, we define a combinatorial object, called weak $B_e$-set, to identify error positions in each row. This method works for arbitrary $e$. The second method, which utilizes algorithms of encoding/decoding a binary sequence of length $n$ and Hamming weight $w$ into/from a number in $\sparenv{\binom{n}{w}}$, only works for the case $e=w$.

Both methods represent each row of our matrices using a symbol of an alphabet of size $Q$ where $\binom{n}{e} \leq Q \leq n^e$. This means that in all previous constructions, we can replace $p^e$ with $Q$. So the guaranteed code size is increased by $p^e/Q\ge (p/n)^e$ times. Or equivalently, the redundancy is reduced by $e\log(p/n)$.
Besides, since the first method uses only one, instead of $e$, functions to encode each row, the decoding process is more efficient.

\begin{definition}\label{dfn_goodset}
  Let $Q>1$ be an integer. A subset $B\subseteq\mathbb{Z}_Q$ is called a \emph{\textbf{weak}} $B_e$-set, where $e<\abs{B}$, if
  $$
  \abs{\mathset{b_1+\cdots+b_{\bar{e}}:~b_1,\ldots,b_{\bar{e}}\in B\text{ are mutually distinct}}}=\binom{\abs{B}}{\bar{e}}, 
  $$
  for each $1\le\bar{e}\le e$.
In other words, the sum $b_1+\cdots+b_{\bar{e}}$ uniquely determines the set $\mathset{b_1,\ldots,b_{\bar{e}}}$, for any distinct $b_1,\ldots,b_{\bar{e}}\in B$. Here, the operation ``$+$" is performed over $\mathbb{Z}_Q$.
\end{definition}

Suppose that $B=\mathset{b_0,\ldots,b_{n-1}}\subseteq\mathbb{Z}_Q$ is a weak $B_e$-set. Without loss of generality, assume $b_0<\cdots <b_{n-1}$. For a sequence $\bm{x}\in\{0,1\}^n$, define $g_{B}(\bm{x})\triangleq\sum_{j=0}^{n-1}b_jx_j\pmod{Q}$. Suppose that $\bm{x}$ suffers $e$ asymmetric errors at positions $j_1,\ldots,j_e$ (where $j_1<\cdots<j_e$) and we receive $\bm{y}$. Then we have
\begin{equation}\label{eq_vefifyBe}
    b_{j_1}+\cdots +b_{j_e}\equiv g_B(\bm{x})-g_B(\bm{y})\pmod{Q}.
\end{equation}
To decode $\bm{x}$ from $\bm{y}$, we need to know the values of $j_1,\ldots,j_e$. By the definition of weak $B_e$-set, we can know $b_{j_1},\ldots,b_{j_e}$ and thus $j_1,\ldots,j_e$, if the value of $g_B(\bm{x})$ is known to us. Note that there are at most $\binom{n}{e}$ possibilities for $\parenv{j_1,\ldots,j_e}$ and for each choice, we only have to verify single modular equation (\Cref{eq_vefifyBe}). Therefore, we can recover $\parenv{j_1,\ldots,j_e}$ efficiently by brute-force search. As a comparison, if we encode $\bm{x}$ into $f_{p,e}(\bm{x})$, then we have to verify $e$ modular equations. As a result, encoding with a weak $B_e$ set is more efficient.

The analysis above implies that all constructions in \Cref{subsection_outer} can be re-established by replacing $f_{p,e}\parenv{\cX_r}$ with $g_B\parenv{\cX_r}$. Here we only show alternative constructions of \Cref{constuction_uniquegeneral} and \Cref{constuction_uniquedecoding2}.
\begin{construction}\label{construction_uniquedecoding4}
  Let $Q>1$ be an integer. Suppose that $B=\mathset{b_0,\ldots,b_{n-1}}\subseteq\mathbb{Z}_Q$ is a weak $B_e$-set. Let $m>1$, $1\le e\le w<n$ be integers.
  \begin{enumerate}[$(1)$]
    \item Let $q\ge\max\mathset{m,Q}$ be the smallest prime. For any $\bm{a}=\parenv{a_0,a_1,\ldots,a_{t-1}}\in\mathbb{Z}_{tQ}\times\mathbb{Z}_{q}^{t-1}$, the code $\cU_{3}$ is defined as
  $$
  \cU_{3}\triangleq\mathset{\cX\in\Sigma_{w}^{m\times n}:\sum_{r=0}^{m-1}g_{B}\parenv{\cX_{r}}\equiv a_0\pmod{tQ},~\sum_{r=0}^{m-1}r^{\ell}g_{B}\parenv{\cX_{r}}\equiv a_{\ell}\pmod{q},\forall\ell\in\sparenv{1,t-1}}.
  $$
    \item If $\gcd\parenv{m,(Q-1)!}=1$, for any $\bm{a}=\parenv{a_1,a_2}\in\mathbb{Z}_{2Q}\times\mathbb{Z}_{m}$, the code $\cU_{3}^\prime$ is defined as
  $$
  \cU_{3}^\prime\triangleq\mathset{\cX\in\Sigma_{w}^{m\times n}:\sum_{r=0}^{m-1}g_{B}\parenv{\cX_{r}}\equiv a_1\pmod{2Q},\sum_{r=0}^{m-1}r\cdot g_{B}\parenv{\cX_{r}}\equiv a_2\pmod{m}}.
  $$
  \end{enumerate}
\end{construction}

\begin{proposition}\label{prop_unique4}
   The code $\cU_{3}$, $\cU_{3}^\prime$ is a $(t,e)$-CAECC,  $(2,e)$-CAECC, respectively. There exists an $\bm{a}$ such that $\rho\parenv{\cU_{3}}\le(t-1)\log(m)+\log(tQ)+t-1$ and there exists an $\bm{a}$ such that $\rho\parenv{\cU_{3}^\prime}\le\log(m)+\log(Q)+1$.
\end{proposition}
This proposition can be proved as \Cref{thm_uniquegeneral} and \Cref{thm_unique2}, just by noticing that given $g_B\parenv{\cX_r}$, $\cX_r$ can be uniquely recovered.

Note that in \Cref{dfn_goodset}, it is required that $b_1,\ldots,b_{\bar{e}}$ are mutually distinct. A quite close concept is the well-known ``$B_e$ set", where it is required that all the sums $b_1+\cdots+b_{e}$ with $b_1\le\cdots\le b_e$ are different. Clearly, a $B_e$ set is a weak $B_e$-set. But the opposite is not necessarily true.

\begin{example}\label{example_comparison}
  It is easy to verify that $B=\mathset{0,1,2,4}$ is a weak $B_2$-set
  in $\mathbb{Z}_6$. Since $2+2=0+4$, $B$ is not a $B_2$ set.
\end{example}

In \cite{Sidon1932}, Sidon introduced $B_2$ sets. In \cite{Singer1938}, Singer constructed optimal $B_2$ sets of size $n$ in $\mathbb{Z}_Q$, where $n-1$ is a prime power and $Q=n^2-n+1$. Following \cite{Singer1938}, Bose and Chowla \cite{Bose1962} gave a construction of $B_e$ sets in $\mathbb{Z}_Q$ for arbitrary $e\ge2$, when $n$ is a prime power and $Q=n^e-1$; or $n-1$ is a prime power and $Q=((n-1)^{e+1}-1)/(n-2)$.
According to \Cref{prop_unique4}, we have $(2,e)$-CAECCs with redundancy at most $\log(m)+\log(n^e-1)+2$ when $n$ is a prime power, or at most $\log(m)+\log((n-1)^{e+1}-1)-\log(n-2)+2$ when $n-1$ is a prime power. In either case, the redundancy is lower than the redundancy guaranteed by \Cref{thm_unique2}.

Known constructions of $B_e$-sets (with as small $Q$ as possible) require that $n$ or $n-1$ is a prime power, which limits the choice of $n$ when it comes to constructing $(t,e)$-CAECCs with redundancy as small as possible. On the other hand, as implied in \Cref{example_comparison}, there might exist a weak $B_e$-set of size larger than those of weak $B_e$-sets given by $B_e$ sets. Therefore, it is possible that there are $(2,e)$-CAECCs with even larger cardinality. In addition, when $e=w$ and $Q=\binom{n}{e}$, the existence of weak $B_e$-sets of size $n$ in $\mathbb{Z}_Q$ enables encoding/decoding algorithms for codes in \Cref{construction_uniquedecoding4} (see \Cref{thm_encoder}). 

\begin{table}[!t]
  \centering
  \caption{weak $B_e$-set $B$ of size $n$ in $\mathbb{Z}_{Q_{n,e}}$, where $2\le e\le\floorenv{n/2}$}\label{tab_goodset}
  \begin{tabular}{c|c|c|c}
    \hline
    \hline
    $e$ & $n$ &$Q_{n,e}$ & $B$\\
    \hline
    $2$& $4$& $6$ & $\mathset{0,1,2,4}$\\
    \hline
    $2$&$5$&$11$&$\mathset{0,1,2,4,7}$\\
    \hline
    $2$&$6$&$19$&$\mathset{0,1,2,4,7,12}$\\
    \hline
    $3$&$6$&$20$ &$\mathset{0,1,2,4,7,13}$\\
    \hline
    $2$&$7$&$28$&$\mathset{0,1,2,4,8,15,20}$\\
    \hline
    $3$&$7$&$39$&$\mathset{0,1,2,4,7,21,32}$\\
    \hline
    $2$&$8$&$40$&$\mathset{0,1,5,7,9,20,23,35}$\\
    \hline
    \hline
  \end{tabular}
\end{table}

Therefore, there is a strong motivation to use weak $B_e$-sets instead of $B_e$ sets to encode rows, as well as to find the smallest $Q$ such that there is a weak $B_e$-set of size $n$ in $\mathbb{Z}_Q$, for given $n$ and $e$. For given $n$ and $2\le e<n$, denote by $Q_{n,e}$ the smallest $Q$ such that there is a weak $B_e$-set of size $n$ in $\mathbb{Z}_Q$.

\begin{lemma}\label{lem_restricte}
  If $\floorenv{n/2}\le e\le n-2$, then $B\subseteq\mathbb{Z}_{Q}$ is a weak $B_e$-set if and only if it is a weak $B_{e+1}$-set in $\mathbb{Z}_Q$. In particular, we have $Q_{n,e+1}=Q_{n,e}$.
\end{lemma}
\begin{IEEEproof}
By definition, it is clear that a weak $B_{e+1}$-set is also a weak $B_e$-set. Now we only need to show the necessity.

Suppose there are two subsets $\mathset{b_1,\ldots,b_{\bar{e}+1}}, \mathset{c_1,\ldots,c_{\bar{e}+1}}\subseteq B$ such that $b_1+\ldots+b_{\bar{e}+1}\equiv c_1+\ldots+c_{\bar{e}+1}\pmod{Q}$, where $1\le \bar{e}\le e$. If $\bar{e}<e$, then $\mathset{b_1,\ldots,b_{\bar{e}+1}}=\mathset{c_1,\ldots,c_{\bar{e}+1}}$ since $B$ is a weak $B_e$-set.

Now suppose $\bar{e}=e$. Since $e\ge\floorenv{n/2}$, we have $2\parenv{e+1}>n$, which implies that the two subsets $\mathset{b_1,\ldots,b_{e+1}}$ and $\mathset{c_1,\ldots,c_{e+1}}$ must intersect. Without loss of generality, suppose $b_{e+1}=c_{e+1}$. Then we have $b_1+\cdots+b_{e}\equiv c_1+\cdots+c_{e}\pmod{Q}$. Since $B$ is a weak $B_e$-set, this implies $\mathset{b_1,\ldots,b_{e}}=\mathset{c_1,\ldots,c_{e}}$. Therefore, $B$ is a weak $B_{e+1}$-set.
\end{IEEEproof}

By \Cref{lem_restricte}, to find $Q_{n,e}$, we can assume $2\le e\le\floorenv{n/2}$. Then it is easy to see that $Q_{n,e}\ge\binom{n}{e}$ by the definition of weak $B_e$-sets. We find by computer searching the smallest $Q$ such that there exists a weak $B_e$-set in $\mathbb{Z}_Q$ for $4\le n\le 7$ and $2\le e\le\floorenv{n/2}$ and $(n,e)=(8,2)$. See \Cref{tab_goodset}. When $n=5$ and $e=4$, \Cref{construction_uniquedecoding4} (2) and \Cref{tab_goodset} implies that there is a $(2,4)$-CAECC with redundancy $\log(m)+\log(11)+1$, while the redundancies guaranteed by \Cref{constuction_uniquedecoding2} and $B_4$ sets are $\log(m)+\log(5^4)+1$ and $\log(m)+\log(341)+1$, respectively.

\begin{example}
\begin{itemize}
  \item In \Cref{example_comparison}, we gave a weak $B_2$ set $B=\mathset{0,1,2,3}$ in $\mathbb{Z}_6$. By \Cref{prop_unique4}, there is a $(2,2)$-CAECC, where $n=4$, with redundancy at most $\log(m)+\log(6)+1$ if none of $2,3$ and $5$ is a factor of $m$. On the other hand, \Cref{thm_unique2} gives a code with redundancy at most $\log(m)+\log(25)+1$ if neither $2$ or $3$ is not a factor of $m$.

  Let $w=e=2$, $n=4$ and $m=7$, by computer searching, we find that when $a_1=0$ and $a_2=5$, $\cU_3^\prime$ has the maximum size $3628$, which implies its redundancy is $7=\ceilenv{\log(7)+\log(6)+1}$.
\end{itemize}
\end{example}

Recall that $Q_{n,e}\ge\binom{n}{e}$. The next theorem implies that the case $Q=\binom{n}{e}$ is of great interest.
In \cite[Section V]{Omer2024ISIT}, the authors gave an efficient encoder for the $(t,1)$-CAECC in \cite[Theorem 1]{Omer2024ISIT} when $n=p$ is a prime number and $m\le n$. Next, we show that, under some conditions, there is an efficient encoder for the $(t,e)$-CAECCs in \Cref{prop_unique4}. For simplicity, we only present the conclusion for $t=2$.

\begin{theorem}\label{thm_encoder}
Let $e=w\le\floorenv{n/2}$ and $Q=\binom{n}{e}$. Suppose that there exists a weak $B_w$ set $B$ in $\mathbb{Z}_Q$. Then there is an efficient encoder for the $(2,e)$-CAECCs $\cU_3$ in \Cref{prop_unique4}, with redundancy $2\floorenv{\log\binom{n}{w}}+\ceilenv{m\parenv{\log\binom{n}{w}-\floorenv{\log\binom{n}{w}}}}$.
\end{theorem}
\begin{IEEEproof}
We first divide our length-$(m-2)\floorenv{\log\binom{n}{w}}$ information sequence into $(m-2)$ disjoint windows, each of length $\floorenv{\log\binom{n}{w}}$. Since each length-$n$ binary sequence with Hamming weight $w$ can be represented as a length-$\ceilenv{\log\binom{n}{w}}$ sequence, each length-$\floorenv{\log\binom{n}{w}}$ binary sequence can be encoded into a length-$n$ sequence with Hamming weight $w$. Now we get $m-2$ length-$n$ sequences of Hamming weight $w$. We put them in the first $m-2$ rows of our codeword $\cX$. Now it remains to show how to get $\cX_{m-2}$ and $\cX_{m-1}$. From the definition of $\cU_3$, we have
\begin{align*}
  g_B(\cX_{m-2})+g_B(\cX_{m-1})\equiv a_0-\parenv{\sum_{i=0}^{m-3}g_B(\cX_i)}\pmod{2Q},\\
  (m-2)g_B(\cX_{m-2})+(m-1)g_B(\cX_{m-1})\equiv a_1-\parenv{\sum_{i=0}^{m-3}ig_B(\cX_i)}\pmod{q}.
\end{align*}
Then by \Cref{lem_uniquetis2}, $g_B(\cX_{m-2})$ and $g_B(\cX_{m-1})$ can be uniquely determined. Since $Q=\binom{n}{w}$ and $B$ is a weak $B_w$ set in $\mathbb{Z}_Q$, $\cX_{m-2}$ and $\cX_{m-1}$ can be uniquely determined. Now we get the codeword $\cX$.

By definition, the redundancy of this encoder is $\ceilenv{m\log\binom{n}{w}}-(m-2)\floorenv{\log\binom{n}{w}}=2\floorenv{\log\binom{n}{w}}+\ceilenv{m\parenv{\log\binom{n}{w}-\floorenv{\log\binom{n}{w}}}}$.
\end{IEEEproof}

\begin{remark}
  Denote $\alpha_n=\log\binom{n}{w}-\floorenv{\log\binom{n}{w}}$ for $n>w$. Then $0<\alpha_n<1$ and the redundancy of the encoder in \Cref{thm_encoder} is $2\floorenv{\log\binom{n}{w}}+\ceilenv{m\alpha_n}$. If $\ceilenv{m\alpha_n}$ is small, the redundancy is near-optimal.
\end{remark}

Although \Cref{thm_encoder} motivates the search for weak $B_e$-sets of size $n$ in $\mathbb{Z}_Q$, where $Q=\binom{n}{e}$, currently we do not how to determine $Q_{n,e}$ theoretically, let alone when it holds that $Q_{n,e}=\binom{n}{e}$. To extend the encoding idea to a more range of values of $w$, we need to look for alternative solutions. Let $Q=\binom{n}{w}$. In \cite{kabal2018combinatorial}, Kabal introduced an efficient way to encode/decode a binary sequence of length $n$ and Hamming weight $w$ into/from an integer in $\cA_Q$. We denote the encoding and decoding function by $g_{enc}\parenv{\cdot}$ and $g_{enc}^{-1}\parenv{\cdot}$, respectively. Replacing $g_B\parenv{\cdot}$ in \Cref{construction_uniquedecoding4}, \Cref{prop_unique4} and \Cref{thm_encoder} with $g_{enc}\parenv{\cdot}$, the following theorem is clear.
\begin{theorem}\label{thm_uniquelast}
  There are $\sparenv{m,(n,w);t,w}$-composite codes with redundancy at most $(t-1)\log(m)+\log\binom{n}{w}+\log(t)+t-1$. Furthermore, there is an efficient encoder/decoder for these codes and the redundancy of this encoder is $t\floorenv{\log\binom{n}{w}}+\ceilenv{m\parenv{\log\binom{n}{w}-\floorenv{\log\binom{n}{w}}}}$.
\end{theorem}

\begin{remark}
  The code in \Cref{thm_uniquelast} is a $(t,e)$-CAECC for all $e$. But when $e<w$, the constructions before \Cref{thm_uniquelast} could provide codes with larger cardinalities.
\end{remark}

\section{List-decodable CAECCs}\label{sec_list}
In this section, we investigate how much we will gain if the requirement for unique-decoding is relaxed.
For any $\cX\in\Sigma_w^{m\times n}$, denote
$$
B_{(t,e)}\parenv{\cX}\triangleq\mathset{\cY\in\{0,1\}^{m\times n}:~\cY\text{ is obtained from }\cX\text{ by a }(t,e)\text{-error}}.
$$
Let $\cE_{t,e}^{m,n}\triangleq\cup_{\cX\in\Sigma_w^{m\times n}}B_{(t,e)}\parenv{\cX}$.

According to \Cref{rmk_exactte}, for each $\cY\in B_{(t,e)}\parenv{\cX}$, we always assume that there are exactly $t$ erroneous rows and in each erroneous row, there are exactly $e$ composite-asymmetric errors. Then the number of words $\cX\in\Sigma_w^{m\times n}$ such that $\cY\in B_{(t,e)}\parenv{\cX}$ is
\begin{equation}\label{eq_ballofY}
n_{0}\triangleq\binom{n-w+e}{e}^t.
\end{equation}
By double-counting, we have
\begin{equation}\label{eq_errorwords}
  \abs{\cE_{t,e}^{m,n}}=\frac{\binom{n}{w}^m\binom{m}{t}\binom{w}{e}^t}{\binom{n-w+e}{e}^t}.
\end{equation}

\begin{definition}
  A set $\cC\subseteq\Sigma_w^{m\times n}$ is called a $(t,e,L)$-composite asymmetric errors ($(t,e,L)$-CAE)  list-decodable code, if for any $\cY\in\cE_{t,e}^{m,n}$, there are at most $L$ $\cX\in\cC$ such that $\cY\in B_{(t,e)}\parenv{\cX}$.
\end{definition}

By \Cref{eq_ballofY}, if $L\ge\binom{n-w+e}{e}^t$, any code is $(t,e,L)$-CAE list-decodable. Therefore, in the rest of this section, it is always assumed that $L<\binom{n-w+e}{e}^t$.

\subsection{Bounds}
The following upper bound is a straightforward generalization of \cite[Theorem 3]{Omer2024ISIT}.
\begin{proposition}
Suppose $L<\binom{n-w+e}{e}^t$. Let $\cC\subseteq\Sigma_w^{m\times n}$ be a $(t,e)$-CAE list-decodable code with list size $L$. Then
  $$
  \abs{\cC}\le L\frac{\binom{n}{w}^{m}}{\binom{n-w+e}{e}^t},
  $$
  and $\rho\parenv{\cC}\ge t\log\binom{n-w+e}{e}-\log(L)$.
\end{proposition}
Although this bound is trivial, in \Cref{subsec_listconstruction}, we will show that when $L$ is much larger than $t$, there are $(t,e,L)$-CAE list-decodable codes with arbitrary length $m$ and redundancy close to this lower bound.

\bigskip
In \cite[Lemma 1]{Omer2024ISIT}, it was proved that $\cC$ can uniquely correct $(t,e)$-CAECCs if and only if $t\le\ed{2e}{\cC}-1$. When $t\ge\ed{2e}{\cC}$, for a given $\cY\in\cE_{t,e}^{m,n}$, there might be more than one codeword $\cX\in\cC$ such that $\cY\in B_{(t,e)}\parenv{\cX}$. Suppose the number of such codewords is $L$. Then we have the following upper bound on $L$.
\begin{theorem}[Plotkin/Johnson-type Bound]\label{thm_johnsonlistsize}
Let $N=\binom{n-w+e}{e}$ and $\cC\subseteq\Sigma_w^{m\times n}$ be a subset with $\ed{2e}{\cC}\ge d>\frac{N-1}{N}t$. Let $\cY\in\cE_{t,e}^{m,n}$. Suppose that there are exactly $s\le t$ rows $\cY_i$ with $\wth{\cY_i}<w$. If there are $L$ $\cX\in\cC$ such that $\cY\in B_{(t,e)}\parenv{\cX}$, then it holds that
$$
L\le\frac{Nd}{d-(s-d)(N-1)}.
$$
\end{theorem}
\begin{IEEEproof}
By assumption, there is a subset $I\subseteq[m]$ with $\abs{I}=s$ such that $\wth{\cY_i}<w$ if and only if $i\in I$. Without loss of generality, assume $I=[s]$.
Suppose $\cX^{(1)},\ldots,\cX^{(L)}\in\cC$ are such that $\cY\in B_{(t,e)}\parenv{\cX^{(k)}}$ for all $1\le k\le L$. Then according to the definition of a $(t,e)$-error, we conclude that $\cX^{(k_1)}_i\ne\cX^{(k_2)}_i$ if and only if $i\in [s]$, for all $1\le k_1\ne k_2\le L$. This implies that a row index $i$ contributes to $\ed{2e}{\cX^{(k_1)},\cX^{(k_2)}}$ only if $i\in[s]$.

For each $i\in[s]$ and $1\le k\le L$, since $\cY_i$ is obtained from $\cX^{(k)}_i$ by replacing at most $e$ $1$'s with $0$'s, we conclude that there are at most $\binom{n-w+e}{e}$ choices for $\cX^{(k)}_i$. On the other hand, since $\cX^{(k)}_i$ is a binary vector of length $n$ and Hamming weight $w$, there are at most $\binom{n}{w}$ choices for $\cX^{(k)}_i$.  Therefore, each $\cX^{(k)}_i$ can be viewed as a symbol in an alphabet $\cA_i$ of size $N$. Let $n_{a,i}$ be the number of $k$'s such that $\cX^{(k)}_i=a$. Then we have $\sum_{a\in\cA_i}n_{a,i}=L$ and hence
\begin{equation}\label{eq_listsize1}
  \sum_{a\in\cA_i}n_{a,i}^2\ge\frac{1}{\abs{\cA_i}}\parenv{\sum_{a\in\cA_i}n_{a,i}}^2=\frac{L^2}{N},
\end{equation}
for each $i\in[s]$.
Since $\ed{2e}{\cC}\ge d$ and $d_H\parenv{\cX^{(k_1)}_i,\cX^{(k_2)}_i}\le 2e$ whenever $i\in[s]$, we have
\begin{equation*}
\begin{aligned}
  L(L-1)d &\le\mathop{\sum}\limits_{1\le k_1\ne k_2\le L}\ed{2e}{\cX^{(k_1)},\cX^{(k_2)}}\\
  &=\sum_{i\in[s]}\sum_{a\in\cA_i}n_{a,i}\parenv{L-n_{a,i}}\\
  &=L\sum_{i\in[s]}\sum_{a\in\cA_i}n_{a,i}-\sum_{i\in[s]}\sum_{a\in\cA_i}n_{a,i}^2\\
  &=sL^2-\sum_{i\in[s]}\sum_{a\in\cA_i}n_{a,i}^2\\
  &\le sL^2-\frac{sL^2}{N},
\end{aligned}
\end{equation*}
where the last inequality follows from \Cref{eq_listsize1}. Note that the requirement $d>\frac{N-1}{N}t$ ensures that $d-(s-d)(N-1)>0$ for all $1\le s\le t$.
Now the proof is completed.
\end{IEEEproof}

Suppose $d\le t<\frac{N}{N-1}d$. \Cref{thm_johnsonlistsize} implies that a $(d-1,e)$-CAECC can also list-decode a $(t,e)$-CAE with a constant list-size.
\begin{corollary}
  Let $N=\binom{n-w+e}{e}$ and $\cC\subseteq\Sigma_w^{m\times n}$ be a subset with $\ed{2e}{\cC}=d>\frac{N-1}{N}t$. For all $d\le t\le m$ and $1\le e\le w$, $\cC$ is a $(t,e)$-CAE list-decodable code with list size at most $\frac{Nd}{d-(t-d)(N-1)}$.
\end{corollary}

As an application of \Cref{thm_johnsonlistsize}, we next show an Elias-type upper bound on the size of uniquely-decodable $(d-1,e)$-CAECCs.
Suppose $\cC\subseteq\Sigma_w^{m\times n}$ with $\ed{2e}{\cC}\ge d>\frac{N-1}{N}t$. Let
$$
T=\mathset{\parenv{\cY,\cX}:\cY\in\cE_{t,e}^{m,n},\cX\in\cC}.
$$
For each $\cY\in\cE_{t,e}^{m,n}$, let $n_{\cC}\parenv{\cY}$ denote the number of codewords $\cX\in\cC$ such that $\parenv{\cY,\cX}\in T$. Then \Cref{thm_johnsonlistsize} asserts that $n_{\cC}\parenv{\cY}\le\frac{d}{d-\frac{N-1}{N}t}$.
\begin{corollary}[Elias-type upper bound]
    With above notations, we have
    $$
    \abs{\cC}\le\frac{d}{d-\frac{N-1}{N}t}\cdot\frac{\binom{n}{w}^m}{\binom{n-w+e}{e}^t}.
    $$
\end{corollary}
\begin{IEEEproof}
We count the size of $T$ in two different ways.
For each $\cX\in\cC$, there are exactly $\binom{m}{t}\binom{w}{e}^t$ $\cY$'s such that $\parenv{\cY,\cX}\in T$. Therefore, $\abs{T}=\abs{\cC}\binom{m}{t}\binom{w}{e}^t$. On the other hand, it is clear that $\abs{T}=\sum_{\cY\in\cE_{t,e}^{m,n}}n_{\cC}\parenv{\cY}$. Combining these two expressions for $\abs{T}$, we conclude that there is a $\cY\in\cE_{t,e}^{m,n}$ with
$$
n_{\cC}\parenv{\cY}\ge\frac{\abs{\cC}\binom{m}{t}\binom{w}{e}^t}{\abs{\cE_{t,e}^{m,n}}}.
$$
Since $n_{\cC}\parenv{\cY}\le\frac{d}{d-\frac{N-1}{N}t}$, we get
\begin{align*}
    \abs{\cC}&\le\frac{d}{d-\frac{N-1}{N}t}\cdot\frac{\abs{\cE_{t,e}^{m,n}}}{\binom{m}{t}\binom{w}{e}^t}\\
    &=\frac{d}{d-\frac{N-1}{N}t}\cdot\frac{\binom{n}{w}^m}{\binom{n-w+e}{e}^t}.
\end{align*}
The last equality follows from \Cref{eq_errorwords}.
\end{IEEEproof}

When $d$ is a constant, this bound is weaker than those in \Cref{thm_uniquelybound1}. When when $d$ and $t$ are proportional to $m$, the Elias-type bound could outperform those bounds in \Cref{thm_uniquelybound1}. Next, we will show this.

To simplify our analysis, fix $n,w$ and $e=w$. Let $t=\tau m$ and $d=\delta m$, where $0<\tau<1$ and $\frac{N-1}{N}\tau<\delta<1$. Given $\tau$ and $\delta$, let $U^S(\delta,m)$ be the sphere-packing upper bound of $(\delta m-1,e)$-CAECCs in \Cref{thm_uniquelybound1} and $U^E(\delta,\tau,m)$ be the Elias-type upper bound (with $t=\tau m$) of $(\delta m-1,e)$-CAECCs. Define
\begin{align*}
    \alpha(\delta)=\limsup_{m\rightarrow\infty}\frac{\log\parenv{U^S(\delta,m)}}{m\log\binom{n}{w}},\\
    \beta(\delta,\tau)=\limsup_{m\rightarrow\infty}\frac{\log\parenv{U^E(\delta,\tau,m)}}{m\log\binom{n}{w}}.
\end{align*}
It is easy to see that $\beta(\delta,\tau)=1-\tau$. To calculate $\alpha(\delta)$,
let $Q=\sum_{l=1}^{\floorenv{e/2}}\binom{w}{l}\binom{n-w}{l}+1$. Since $\sum_{s=0}^{\floorenv{\delta m-1/2}}\binom{m}{s}\parenv{Q-1}^s\ge Q^{m\parenv{H_Q(\delta/2)-o(1)}}$, where $H_Q(\cdot)$ is the $Q$-ary entropy function, by \Cref{thm_uniquelybound1}, we have $\alpha(\delta)=1-\log(Q)/\log\binom{n}{w}\cdot H_Q(\delta/2)$. If $\tau>\log(Q)/\log\binom{n}{w}\cdot H_Q(\delta/2)$, it holds that $\beta(\delta,\tau)<\alpha(\delta)$. As an example, let $n=4$, $e=w=2$. Then $Q=5$ and $N=6$. Let $\delta=1.01\cdot\frac{N-1}{N}\tau$. Then $\alpha(\delta)=1-\log(5)/\log(6))\cdot H_5\parenv{\frac{5.05}{6}\tau}$. We plot functions $\beta(\delta,\tau)$ and $\alpha(\delta)$ in \Cref{fig_bounds}. It is to see that when $\tau\ge 0.4$, the asymptotic Elias-type upper bound outperforms the asymptotic Sphere-packing bound.

\begin{figure}[!t]
    \centering
    \includegraphics[width=0.5\textwidth]{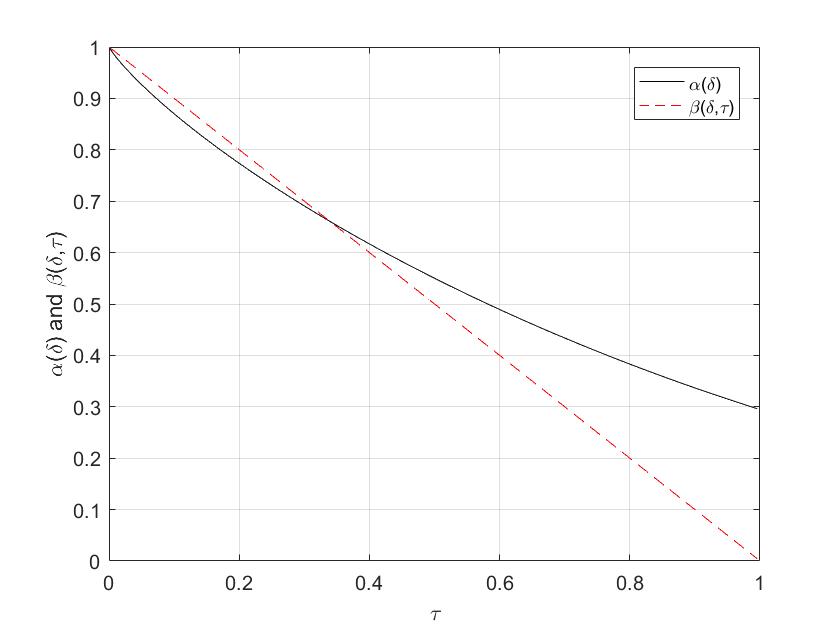}
    \caption{Comparison between Sphere-packing bound and Elias-type upper bound when $n=4$, $e=w=2$ and $\delta=\frac{101}{120}\tau$.}
    \label{fig_bounds}
\end{figure}

\subsection{Constructions}\label{subsec_listconstruction}
In this section, we give three constructions of list-decodable $(t,e)$-CAECCs with constant list-size $L$. In the first two constructions, we consider general $t$. The resulted $L$ is factorial in $t$ or an exponential function in $t$. Recall that we assume $t,n,w$ and $e$ are constants with respect to $m$. This implies that $L$ is a constant. Moreover, the redundancies of the first two codes are also constants.

In the third construction, we focus on $t=3$. We obtain a list-decodable $(3,e)$-CAECC with list-size $2$ and the redundancy is $\log(m)+O(1)$. It is worth noting that the redundancy differs by a constant from the lower bound on the smallest redundancy of \emph{uniquely} decodable $(t,e)$-CAECCs implied by \Cref{thm_uniquelybound1}.
\begin{construction}\label{constuction_list-decoding1}
  Let $m,n>1$, $1\le w<n$, $1\le e\le w$ and $1\le t<m$. Let $p$ be the smallest prime in $\sparenv{n,2n}$ and $q$ is the smallest prime between $\sparenv{p^e,2p^e}$. Let $0\le a_1,\ldots, a_t<q$. The code $\cL_{1}$ is defined as
  $$
  \cL_{1}\triangleq\mathset{\cX\in\Sigma_{w}^{m\times n}:~\sum_{i=0}^{m-1}f_{p,e}\parenv{\cX_i}^{\ell}\equiv a_{\ell}\pmod{q},\forall\ell\in[1,t]}.
  $$
\end{construction}

\begin{theorem}\label{thm_list1}
 When $q>t$, the code $\cL_{1}$ is a $(t,e,L)$-CAE list-decodable code with $L=t!$. There exist some $a_1,\ldots,a_t$ such that $\rho(\cL_1)\le t\log\parenv{q}\in\sparenv{et\log(p),et\log(p)+t}$.
\end{theorem}
\begin{IEEEproof}
Let $\cY$ be obtained from a codeword $\cX$ by a $(t,e)$-CAE. Comparing $\cY$ with $\cX$, we can know which rows suffered errors. Suppose that the positions in $f\parenv{\cX}$ suffered erasures are $i_1,\ldots,i_t$. Let $a_{\ell}^{\prime}=\parenv{a_{\ell}-\sum_{i\ne i_1,\ldots,i_{t}}f_{p,e}\parenv{\cX_{i}}^{\ell}}\pmod{q}$. Then we have
\begin{equation*}
  \left\{
  \begin{array}{c}
    f_{p,e}\parenv{\cX_{i_1}}+\cdots +f_{p,e}\parenv{\cX_{i_s}}\equiv a_1^{\prime}\pmod{q},\\
    \vdots \\
    f_{p,e}\parenv{\cX_{i_1}}^t+\cdots +f_{p,e}\parenv{\cX_{i_t}}^t\equiv a_t^{\prime}\pmod{q}.
  \end{array}
  \right.
\end{equation*}
Now by \Cref{lem_pre1}, one can recover the multiset $\mathset{f_{p,e}\parenv{\cX_{i_1}},\ldots,f_{p,e}\parenv{\cX_{i_s}}}$. There are at most $s!\le t!$ ways to order the multiset $\mathset{f_{p,e}\parenv{\cX_{i_1}},\ldots,f_{p,e}\parenv{\cX_{i_s}}}$. By \Cref{fact_pre1} and \Cref{lem_pre2}, one can obtain at most one candidate $\cX^\prime$ of $\cX$ for each order. Now the proof is completed.
\end{IEEEproof}

\begin{lemma}\label{lem_numofsolutions1}
  Let $q>1$ be an integer, $0<a_1,\ldots,a_k<q$ and $0\le b<q$. Consider the following congruent equation
  \begin{equation}\label{eq_numberofsolutions1}
  a_1x_1+\cdots+a_kx_k\equiv b\pmod{q}.
  \end{equation}
  Let $d=\gcd\parenv{a_1,\ldots,a_k,q}$. Then \Cref{eq_numberofsolutions1} has a solution if and only if $d\mid b$. Furthermore, when $d\mid b$, the number of solutions in $\mathbb{Z}_q^k$ to \Cref{eq_numberofsolutions1} is $dq^{k-1}$.
\end{lemma}

\begin{construction}\label{construction_list-decoding2}
  Let $m,n>1$, $1\le w<n$, $1\le e\le w$ and $1\le t<m$. Let $p$ be the smallest prime in $\sparenv{n,2n}$. Let $0\le a<p^e$. The code $\cL_2$ is defined as
  $$
  \cL_2\triangleq\mathset{\cX\in\Sigma_{w}^{m\times n}:~\sum_{i=0}^{m-1}f_{p,e}\parenv{\cX_i}=a\pmod{p^e}}.
  $$
\end{construction}

\begin{theorem}
 The code $\cL_2$ a $(t,e,L)$-CAE list-decodable code with $L=p^{e(t-1)}$. There exists some $a$ such that $\rho(\cL_2)\le e\log(p)\in[e\log(n),e\log(n)+1]$.
\end{theorem}
\begin{IEEEproof}
Let $\cY$ be obtained from a codeword $\cX$ by a $(t,e)$-CAE. Comparing $\cY$ with $\cX$, we can know which rows suffered errors. Suppose that the positions in $f\parenv{\cX}$ suffered erasures are $i_1,\ldots,i_t$. Let $a^{\prime}=\parenv{a-\sum_{i\ne i_1,\ldots,i_{t}}f_{p,e}\parenv{\cX_{i}}}\pmod{p^e}$. Then we have
\begin{equation*}
    f_{p,e}\parenv{\cX_{i_1}}+\cdots +f_{p,e}\parenv{\cX_{i_s}}\equiv a^{\prime}\pmod{p^e}.
\end{equation*}
By \Cref{lem_numofsolutions1}, this equation has $p^{e(s-1)}\le p^{e(t-1)}$ solutions. For each solution, one can obtain at most one candidate $\cX^\prime$ of $\cX$.
\end{IEEEproof}

\begin{remark}
  It is clear that $\rho(\cL_2)<\rho(\cL_1)$. But when $t$ is small, it could be that $t!<p^{e(t-1)}$.
\end{remark}

Next, we turn to the case $t=3$ and a much smaller list-size $L$. First, we need the following lemma.
\begin{lemma}\label{lem_listtis3}
Let $Q>1,m>1$ be two integers and $0\le i<j<k< m$. Suppose that $i,j,k$ are known. Let $u$ be a given integer. Suppose that $(x_1,x_2,x_3)=(a,b,c)$ is a solution to the following equation
\begin{equation}\label{eq_constraint_t3}
  ix_1+jx_2+kx_3=u,
\end{equation}
where $0\le a,b,c<Q$. Then given the multiset $\multiset{a,b,c}$, vector $(a,b,c)$ can be determined within up to two options.
\end{lemma}
\begin{IEEEproof}
Suppose $\multiset{\beta_1,\beta_2,\beta_3}=\multiset{\alpha_1,\alpha_2,\alpha_3}=\multiset{a,b,c}$, such that
$i\alpha_1+j\alpha_2+k\alpha_3= i\beta_1+j\beta_2+k\beta_3$.
Then we have
\begin{equation}\label{eq_t3_1}
  i(\alpha_1-\beta_1)+j(\alpha_2-\beta_2)+k(\alpha_3-\beta_3)=0.
\end{equation}
If $\alpha_l=\beta_l$ for some $1\le l\le 3$, then similar to the proof of \Cref{thm_uniquegeneral} (when $t=2$), we can show that $\beta_1=\alpha_1$, $\beta_2=\alpha_2$ and $\beta_3=\alpha_3$. Now assume $\beta_l\ne\alpha_l$ for all $1\le l\le 3$. This means that $\alpha_1,\alpha_2,\alpha_3$ are mutually distinct. Then by \Cref{eq_t3_1}, we conclude that exactly one of $\alpha_l-\beta_l$ is positive or exactly one of $\alpha_l-\beta_l$ is negative. There are six cases:
\begin{enumerate}[$(1)$]
  \item $\alpha_1-\beta_1>0$, $\alpha_2-\beta_2<0$ and $\alpha_3-\beta_3<0$;
  \item $\alpha_1-\beta_1<0$, $\alpha_2-\beta_2<0$ and $\alpha_3-\beta_3>0$;
  \item $\alpha_1-\beta_1<0$, $\alpha_2-\beta_2>0$ and $\alpha_3-\beta_3>0$;
  \item $\alpha_1-\beta_1>0$, $\alpha_2-\beta_2>0$ and $\alpha_3-\beta_3<0$;
  \item $\alpha_1-\beta_1<0$, $\alpha_2-\beta_2>0$ and $\alpha_3-\beta_3<0$;
  \item $\alpha_1-\beta_1>0$, $\alpha_2-\beta_2<0$ and $\alpha_3-\beta_3>0$.
\end{enumerate}

Cases (1)--(4) can be handled in the same way. Hence, we only discuss case (1) here. Then \Cref{eq_t3_1} is equivalent to $i(\alpha_1-\beta_1)=j(\beta_2-\alpha_2)+k(\beta_3-\alpha_3)$. Since $j\ge i+1$ and $k\ge i+2$, we have $i(\alpha_1-\beta_1)\ge i(\beta_2-\alpha_2+\beta_3-\alpha_3)+(\beta_2-\alpha_2)+2(\beta_3-\alpha_3)$, which implies
$i\parenv{\alpha_1+\alpha_2+\alpha_3-\beta_1-\beta_2-\beta_3}\ge \beta_2-\alpha_2+2(\beta_3-\alpha_3)>0$. Then it follows that $\alpha_1+\alpha_2+\alpha_3>\beta_1+\beta_2+\beta_3$, which contradicts the assumption $\multiset{\beta_1,\beta_2,\beta_3}=\multiset{\alpha_1,\alpha_2,\alpha_3}$.

For case (5), there are only two possibilities:
\begin{enumerate}[$(i)$]
  \item $\alpha_1<\alpha_3<\alpha_2$, then $\beta_1=\alpha_3$, $\beta_2=\alpha_1$ and $\beta_3=\alpha_2$; or
  \item $\alpha_3<\alpha_1<\alpha_2$, then $\beta_1=\alpha_2$, $\beta_2=\alpha_3$ and $\beta_3=\alpha_1$.
\end{enumerate}
For case (6), multiply $-1$ at both sides of \Cref{eq_t3_1} and we get case (5).

Now we describe how to determine vector $(a,b,c)$ within up to two options. Firstly, find a permutation $\parenv{\alpha_1,\alpha_2,\alpha_3}$ of $\multiset{a,b,c}$ that satisfies \Cref{eq_constraint_t3}. If at least two of $a,b,c$ are equal, then $\parenv{a,b,c}=\parenv{\alpha_1,\alpha_2,\alpha_3}$. When $a,b,c$ are mutually distinct, if $\parenv{\alpha_1,\alpha_2,\alpha_3}$ satisfies one of the following conditions
\begin{itemize}
  \item $\alpha_1<\alpha_3<\alpha_2$,
  \item $\alpha_3<\alpha_1<\alpha_2$,
  \item $\alpha_2<\alpha_1<\alpha_3$
  \item $\alpha_2<\alpha_3<\alpha_1$,
\end{itemize}
then apart from $\parenv{\alpha_1,\alpha_2,\alpha_3}$, there are at most one permutation of $\multiset{a,b,c}$ that satisfies \Cref{eq_constraint_t3}. For other cases of $\parenv{\alpha_1,\alpha_2,\alpha_3}$, we have $\parenv{a,b,c}=\parenv{\alpha_1,\alpha_2,\alpha_3}$.
\end{IEEEproof}

\begin{construction}[$t=3,L=2$]\label{constuction_listtis3}
  Let $m,n>1$, $1\le w<n$, $1\le e\le w$. Let $p$ be the smallest prime in $\sparenv{n,2n}$ and $q$ be the smallest prime in $\sparenv{p^e,2p^e}$.
  For any $\bm{a}=\parenv{a_1,a_2,a_3,a_4}\in\mathbb{Z}_{q}^3\times\mathbb{Z}_{3mp^e}$, the code $\cL_{3}$ is defined as
  $$
  \cL_{3}\triangleq\mathset{\cX\in\Sigma_{w}^{m\times n}:~\sum_{r=0}^{m-1}f_{p,e}\parenv{\cX_r}^{l}\equiv a_{l}\pmod{q},\forall l\in[1,3],\sum_{r=0}^{m-1}r\cdot f_{p,e}\parenv{\cX_{r}}\equiv a_4\pmod{3mp^e}}.
  $$
\end{construction}

\begin{theorem}\label{thm_listtis3}
 Suppose $q>3$. The code $\cL_3$ is $(3,e,2)$-CAE list-decodable. There is some $\bm{a}$ such that $\rho\parenv{\cL_3}\le\log(m)+\log(3p^e)+3\log(q)$.
\end{theorem}
\begin{IEEEproof}
Let $\cY$ be obtained from a codeword $\cX$ by a $(3,e)$-CAE. Comparing $\cY$ with $\cX$, we can know which rows suffered errors.

Suppose that $\cX_i,\cX_j,\cX_k$ suffered errors, where $0\le i<j<k<m$. According to \Cref{lem_pre1}, we can recover the multiset $\multiset{f_{p,e}\parenv{\cX_i},f_{p,e}\parenv{\cX_j},f_{p,e}\parenv{\cX_k}}$ by the three constraints $\sum_{r=0}^{m-1}f_{p,e}\parenv{\cX_r}^{\ell}\equiv a_{\ell}\pmod{q}$, where $1\le\ell\le3$. Then by \Cref{lem_listtis3} and the constraint $\sum_{r=0}^{m-1}r\cdot f_{p,e}\parenv{\cX_{r}}\equiv a_4\pmod{3mp^e}$, the tuple $\parenv{f_{p,e}\parenv{\cX_i},f_{p,e}\parenv{\cX_j},f_{p,e}\parenv{\cX_k}}$ can be determined up to two options.
For each option, there are at most one codeword corresponding to it.
\end{IEEEproof}

\bigskip
\bigskip
\section{conclusion}\label{sec_conclusion}
In this paper, we improved previous constructions of $\sparenv{m,(n,w);t,e}$-composite codes when $m$ is large while $t,n,w$ and $e$ are constants. We also investigate the list-decodability of $\sparenv{m,(n,w);t,e}$-composite codes. We showed that when the list-size $L$ is large compared to $t$, i.e., $L=t!$ or $L=p^{e(t-1)}$, there are codes with constant redundancy. When $t=3$ and $L=2$, we constructed a list-decodable code with redundancy $\log(m)+O(1)$. This introduces the following interesting problem: is there (and how to find) a function $f_{n,w,e}(t)$ such that when $L>f_{n,w,e}(t)$, there are $(t,e,L)$-CAE list-decodable codes with constant redundancy?

\appendices
\section{Proof of \Cref{lem_pre2}}\label{appendix_pfpre2}
\begin{IEEEproof}
Assume that the error positions are $j_1,\ldots,j_e$, where $0\le j_1<\cdots<j_e<n$. To recover $\bm{x}$, it sufficient to recover the vector $\parenv{j_1,\ldots,j_e}$. To that end, let $a_{\ell}=\parenv{s_{\ell}^p(\bm{x})-s_{\ell}^p(\bm{y})}$ for all $1\le\ell\le e$. Then we have
\begin{equation*}
    \left\{
    \begin{array}{c}
     j_1+\cdots+j_e\equiv a_1\pmod{p}\\
     j_1^2+\cdots+j_e^2\equiv a_2\pmod{p}\\
     \vdots\\
     j_1^e+\cdots+j_e^e\equiv a_e\pmod{p}
    \end{array}
    \right..
  \end{equation*}
  According to \Cref{lem_pre1}, the vector $\parenv{j_1,\ldots,j_s}$ can be uniquely determined. Now the proof is completed.
\end{IEEEproof}

\bigskip

\bibliographystyle{IEEEtran}
\bibliography{ref}
\end{document}